%% file: main-arxiv.tex
\documentclass[11pt]{article}
\usepackage{times} % colt font

\input{packages_arxiv.tex}

\input{theorems_std.tex}
\input{macros.tex}

\PassOptionsToPackage{algo2e,ruled}{algorithm2e}
\RequirePackage{algorithm2e}

\newcommand{\sech}[0]{\operatorname{sech}}
\newcommand{\prl}[0]{\parallel}

\definecolor{purple}{rgb}{0.5, 0.0, 0.5}
\newcommand{\hlnote}[1]{\textcolor{red}{\textbf{HL:} #1}}
\newcommand{\fnote}[1]{\textcolor{purple}{\textbf{FK:} #1}}

\newcommand{\CP}[0]{C_{\textup{P}}}

\newcommand{\CMLS}[0]{C_{\textup{MLS}}}

\newcommand{\BC}[0]{\{\pm 1\}^n}
\newcommand{\sBC}[0]{\sum_{\si\in \{\pm 1\}^n}}

\newcommand{\Grid}[0]{\textup{Grid}_{L,\eta}^d}

\newcommand{\ZJh}[0]{Z_{J,h}}
\newcommand{\Zjoint}[0]{Z_{J,h}^{\textup{joint}}}
\newcommand{\Zproj}[0]{Z_{J^\prl,J^\perp,h}}
\newcommand{\pproj}[0]{p_{J^\prl,J^\perp,h}}
\newcommand{\Pproj}[0]{P_{J^\prl,J^\perp,h}}
\newcommand{\isingexp}[2]{\rc 2 \an{\si, #1 \si} + \an{#2, \si}}
\newcommand{\ising}[2]{\exp\pa{\isingexp{#1}{#2}}}

\newcommand{\eJhs}[0]{\exp\pa{\isingexp{J}{h}}}
\newcommand{\Jpall}[0]{J^\perp_{\textup{all}}}
\newcommand{\Pst}[0]{P^{\textup{st}}}
\newcommand{\pst}[0]{p^{\textup{st}}}
\newcommand{\op}[0]{\textup{op}}

\begin{document}

%\title{Sampling and partition function estimation for approximately low-rank Ising models}
\title{Sampling Approximately Low-Rank Ising Models: \\ MCMC meets Variational Methods}
%\title{Sampling from Ising models with a few large eigenvalues}
%Fred: maybe "spike eigenvalues" or "spectral outliers" would be more detailed and sound better than low rank
%or "approximately low rank"?
%Another option: something like "bounded spectral bulk"?

\author{Frederic Koehler\thanks{Department of Computer Science, Stanford University. Supported in part by NSF award CCF-1704417, NSF Award IIS-1908774, and N. Anari's Sloan Research Fellowship.}, \quad Holden Lee\thanks{Department of Mathematics,
Duke University.
%Email: \texttt{holden.lee@duke.edu}
}, \quad Andrej Risteski\thanks{ Machine Learning Department,
Carnegie Mellon University.}}

\date{\today}
\maketitle
\begin{abstract}
    We consider Ising models on the hypercube with a general interaction matrix $J$, and give a polynomial time sampling algorithm when all but $O(1)$ eigenvalues of $J$ lie in an interval of length one, a situation which occurs in many models of interest. This was previously known for the Glauber dynamics when \emph{all} eigenvalues fit in an interval of length one; however, a single outlier can force the Glauber dynamics to mix torpidly. Our general result implies the first polynomial time sampling algorithms for low-rank Ising models such as Hopfield networks with a fixed number of patterns and Bayesian clustering models with low-dimensional contexts, 
    and greatly improves the polynomial time sampling regime for the antiferromagnetic/ferromagnetic Ising model with inconsistent field on expander graphs. It also improves on previous approximation algorithm results based on the naive mean-field approximation in variational methods and statistical physics. 
    
    Our approach is based on a new fusion of ideas from the MCMC and variational inference worlds. As part of our algorithm, we define a new nonconvex variational problem which allows us to sample from an exponential reweighting of a distribution by a negative definite quadratic form, and show how to make this procedure provably efficient using stochastic gradient descent. On top of this, we construct a new simulated tempering chain (on an extended state space arising from the Hubbard-Stratonovich transform) which overcomes the obstacle posed by large positive eigenvalues, and combine it with the SGD-based sampler to solve the full problem. 
\end{abstract}

%TODO:
%\begin{enumerate}
%    \item Add introduction.
%    \item \hlnote{Add more details to sampling section.}
%    \item Import the negative definite writeup.
%    \item Show how to do partition function estimation and sampling with negative definite components.
%    \item Add examples/applications. \fnote{mostly done, clean up appendix.}% \fnote{two applications to classic models: SK Model with Ferromagnetic Interaction \cite{chen2014mixed,talagrand2010mean}. Hopfield model with a fixed number of patterns \cite{gentz1999fluctuations,talagrand2010mean}.} \fnote{add antiferromagnetic/max-cut ising model on random $d$-regular graph} \fnote{clean up GMM and Contextual SBM example}
    %Q: can you do importance sampling with original convex relaxation?
%    \item \fnote{remember to add comparison to $\|J\|_F^2$ condition from variational inference results somewhere (by considering the interval $[-1/2,1/2]$ as the bulk, we should reproduce similar results but now sampling?)}
    
%    \fnote{pick one of $x$ and $\sigma$ to denote spin} \hlnote{Let's use $\si$ because I used $x_i$ to denote columns of $X$} %fixed
%\end{enumerate}

%Refs: \cite{eldan2020spectral,anari2021aentropic}, \cite{jain2019mean}

\input{intro}
\input{overview}

\section{Applications}\label{sec:applications}
Our results specialize to give new sampling guarantees for a many models of interest.  All of these are Ising models, so in each application we will describe the particular interaction matrix which arises and the resulting runtime guarantee. In all of the applications, the behavior in the presence of an external field $h \in \mathbb{R}^n$ is of interest (for example, in the Hopfield network to preferentially weight the distribution towards a particular memory) and %our results automatically apply to 
we automatically handle this case. 

\paragraph{Hopfield Network with a fixed number of patterns.} The Hopfield network is a neural model of associative memory 
%driven by a Hebbian learning rule 
(\cite{hopfield1982neural}, see also \cite{pastur1977exactly,pastur1978theory,little1974existence}) which has been hugely influential and extensively studied. In particular, for rigorous mathematical results see the textbooks by \cite{bovier2012mathematical,talagrand2010mean}.
Formally, given patterns $\eta_1,\ldots,\eta_m \in \{\pm 1\}^n$ the Hopfield network at inverse temperature $\beta$ is the Ising model with interaction matrix
$J = \frac{\beta}{2n} \sum_{v = 1}^m \eta_v \eta_v^\top$.
This is thought of as a ``Hebbian'' learning rule because for each memory $\eta_v$ and neurons (coordinates) $i$ and $j$, the term $(\eta_v)_i (\eta_v)_j$ is positive if $(\eta_v)_i = (\eta_v)_j$ and negative otherwise. Therefore if $J$ is thought of as the ``wiring'' of the neurons, then for each pattern all of the neurons which ``fire together," i.e., have the same spin, are ``wired together''.
%, matching the idiom that ``neurons which fire together wire together''. 

Most of the interest in this model has been in the case of low/zero-temperature, which means the parameter $\beta$ is large. 
Glauber dynamics (Gibbs sampling) has long been considered as a natural dynamics for the Hopfield network. Informally, the patterns stored in the network serve as ``attractors'' which trap the dynamics. This is interesting as %it can be thought of as a 
in a sense it means the network exhibits memory; however, from the sampling perspective this means that the vanilla Glauber dynamics are not expected to mix in the most interesting regime of this model.

When the number of patterns $m$ is fixed (a regime which has been rigorously studied in e.g., \cite{gentz1999fluctuations,bovier2012mathematical,talagrand2010mean}), we obtain the first polynomial time sampling algorithm for the Gibbs measure of this model that works for any fixed $\beta > 0$. Based on the rigorous results in this model (see \cite{bovier2012mathematical,talagrand2010mean}), when each pattern is independently sampled $\eta_i \sim \mathsf{Uniform} (\{\pm 1\}^n)$ and $\beta > 1$ the distribution will be almost entirely supported on $2m$ clusters corresponding to each of the patterns $\{\pm \eta_i\}_{i = 1}^m$ and so ordinary Glauber dynamics will not mix rapidly. (This should not be too difficult to formally prove given their results, though we did not do this.) Note that our sampling results apply to arbitrary patterns $\eta_i$, not just the commonly studied case where the patterns are uniformly random from the hypercube.

\paragraph{Antiferromagnetic and Ferromagnetic Ising Model on expanders and random graphs.} Suppose that $A$ is the adjacency matrix of a graph; then the \emph{antiferromagnetic Ising model} at inverse temperature $\beta$ has interaction matrix
$J = -\beta A$.
It is known that for worst-case graphs of maximum degree $d$, that polynomial time sampling is only possible for $\beta = O(1/d)$ (\cite{sly2012computational}, in fact the precise threshold is known as a function of $d$). However, this should be far from tight in other cases of interest, such as on a uniformly random $d$-regular graph: in this model, it is known that the symmetry breaking phase transition is at scaling $\beta = \Theta(1/\sqrt{d})$ (see \cite{coja2020ising} and references within) and we would expect the sampling regime of the model to be similar. 
%when $A$ is a spectral expander and we can prove it as a consequence of our main result. 

Based on our main result, we can indeed recover the correct scaling in the random $d$-regular graph setting, as a special case of a much more generic result about spectral expanders.
Let $\lambda = \max\{|\lambda_2(A)|, |\lambda_n(A)|\}$; then our results give a polynomial time sampler whenever $\beta d = O(\log n)$ (so that our algorithm is polynomial time) %\hlnote{Do you mean $\be d$ here?} %fixed
and provided
$\beta \lambda < 1$.
For example, in the case of a \emph{Ramanujan graph} of degree $d$ we have $\lambda \le 2\sqrt{d - 1}$ and so we can sample in polynomial time whenever
$\beta < \frac{1}{2\sqrt{d - 1}}$,
which is a dramatic improvement over $O(1/d)$. Because of Friedman's Theorem, we know the same result holds for the a uniformly random $d$-regular graph since it will be almost-Ramanujan \citep{friedman2008proof}. Note that it is the presence of the ``trivial'' eigenvalue $\lambda_1$ which prevents the result from being deduced from the pre-existing works (e.g., \cite{eldan2020spectral}) which can handle related models (diluted $d$-regular SK model) without outlier eigenvalues.  Our result also applies analogously if there are a couple of outlier eigenvalues, e.g., on bipartite expanders.  

A completely analogous consequence of our theory is for the case of \emph{ferromagnetic} Ising models on expanders, where we have $J = \beta A$. In this case, the famous result of Jerrum and Sinclair \citep{jerrum1993polynomial} proves that sampling is possible when the external field $h$ is \emph{consistent} i.e., $h_i \ge 0$ for all $i$. However, when the signs of the external fields $h_i$ are allowed to disagree, sampling from the ferromagnetic Ising model is \#BIS-Hard \citep{goldberg2007complexity}. So our result also implies sampling algorithms for the ferromagnetic Ising model with inconsistent external field on expanders up to larger inverse temperatures than were previously known. 

%Finally, in the case of bipartite expanders (with either ferromagnetic or antiferromagnetic, $A$ will have both a single large positive and a single large negative eigenvalue. Again, our sampling algorithm allows for sampling 
%\paragraph{Ferromagnetic Ising Model on expanders and with inconsistent fields.} TODO: emphasize that Jerrum-Sinclair does not apply when external fields are inconsistent, but our result does. 

%\paragraph{Ising model on bipartite expanders.} This application gives an intuitive example where both a positive and a negative spike eigenvalues are present.  

\paragraph{Sherrington-Kirkpatrick Model with Ferromagnetic Interaction.} The Sherrington-Kirkpatrick model is one of the most famous spin glass models, and the SK model with ferromagnetic interactions is a natural variant which exhibits a combination of ferromagnetic and spin glass behaviors---see e.g., \cite{chen2014mixed,comets1999sherrington,talagrand2010mean} for rigorous probabilistic analysis of this model. The interaction matrix $J$ is given by
$J_{ij} = \frac{\beta_1}{n}  + \beta_2 W_{ij}$
where $W$ is a matrix sampled from the Gaussian Orthogonal Ensemble (so $W_{ij} \sim N(0,1/n)$). Since $\|W\|_{\op} \le 2(1 + o(1))$ with high probability by classical results in random matrix theory \citep{anderson2010introduction}, we are able to sample in polynomial time from this model for any fixed $\beta_1$, as long as $\beta_2 < 1/4$.

\paragraph{Posterior in Low-Dimensional Gaussian Mixture Model.} A basic clustering problem in Bayes\-ian statistics is posterior inference in the two-component (symmetric) Gaussian mixture model. More specifically, we will consider that we have data points $b_1,\ldots,b_n \in \mathbb{R}^p$ and we want to sample from the posterior under the following Bayesian model: $u \sim N(0,I_p/p)$, $v \sim \mathsf{Uniform} (\{\pm 1\}^n)$ are the latent cluster assignments and independently 
$b_i \sim N(v_i \sqrt{\mu/n}\ u, I_p/p)$.
In other words, we posit that the data points were generated by a balanced mixture of two spherical Gaussians with means $\pm \sqrt{\mu/n}\ u$ and $u$ itself is sampled from a Gaussian distribution. (For simplicity, we assumed that the data is scaled and centered so that the variance of the components is $I_p/p$; the scalings here are chosen in part to maintain consistency with the next example.)  In this case, the posterior on the cluster assignments $v$ is given by $p(v \mid b) \propto \exp\left(\frac{p \mu}{2n(1 + \mu)} \langle vv^\top , BB^\top  \rangle\right)$ 
where $B \in \R^{n \times p}$ is the matrix with rows $b_i$. (See Appendix~\ref{a:example} for the derivation.) Note that this is an Ising model with
$ J = \frac{p \mu}{2n(1 + \mu)} BB^\top  $
and the rank of $J$ is at most $p$. Hence, our main result lets us sample from this distribution (posterior in the Gaussian Mixture Model) in polynomial time in fixed dimension $p$. In the case of a balanced mixture, the posterior will always be bimodal due to the symmetry of swapping the two cluster assignments, and so Glauber dynamics would not be expected to mix. (Also, our algorithms works for general data points $b_1,\ldots,b_n$ in which case the posterior can be an arbitrary positive semidefinite Ising model of rank $p$ --- in particular, it could be a Hopfield network and have even more than two modes.) In fact, the Hubbard-Stratonovich transform and our algorithm as a whole has a natural interpretation in terms of searching over the latent vector $u$ in this case (see Appendix~\ref{app:equiv}).
%Note that the $\mu$ is exactly the variable appearing in the Hubbard-Stratonovich transform, so algorithm has the natural interpretation of iterating over a grid of $\mu\in \R^p$ (see Appendix~\ref{app:equiv}). \fnote{this $\mu$ is not $\mu$ as defined above, which is just a number?}
Finally, we note that this example can be easily generalized to assymetric mixture (mixing weights not 50/50); this just changes the prior, which results in an external field in the (Ising model) posterior. 

\begin{rem}
Importantly, the posterior sampling result we establish does not rely on the data being a typical sample from the posited Bayesian model. This is useful because in many machine learning and statistics applications the data is not exactly generated from the posited model, and nevertheless sampling from the posterior is very useful. On the other hand, if the data is indeed generated from the model (i.e., well-specified) then posterior sampling lets us compute the Bayes-optimal estimator of quantities of interest, e.g., compute $\Pr(v_i = v_j \mid B)$ in the GMM example which is the Bayes-optimal estimate of $1(v_i = v_j)$, the indicator that $i$ and $j$ are from the same component.  
%% May not be true in as strong as a sense is hoped.
%If we assume the data is actually generated from the model, then we can show the naive mean-field approximation becomes accurate as we obtain more data, and so for the posterior and so give a faster algorithm for approximately sampling from the posterior --- see Appendix~\ref{apdx:?}. 
\end{rem}
%\fnote{can easily modify to be an unbalanced mixture, also should add in scaling for variance?}
%\fnote{add example of low-dimensional GMM first?}
%\fnote{gap they establish is $1/\gamma^{3/2}$, is this gonna be useless if $p = O(1)$? Maybe we can still do $p = o(n)$? Also minor fixme: our $B$ is the transpose of theirs}
%\fnote{when $p = 1$ and $\lambda = 0$, threshold for weak recovery in one dimension is not $\mu^2/n > 1$ but more like $\mu/n > 0$? because for 1d mixture of gaussians, we need a constant separation for clustering (as opposed to recovering the mean). it's indeed a different regime, should include 1d experiment to illustrate}

\paragraph{Posterior in Low-Dimensional Contextual SBM.}
The contextual stochastic block model \citep{deshpande2018contextual} is a more complex version of the previous GMM model in which the cluster structure is also reflected in the community structure of a graph. We consider the low-dimensional version of this model where the dimension of the contexts $p$ is small---this is morally related to, but different from, the spiked Wishart model with side information, see e.g., \cite{montanari2021estimation}. %, though not the same.
For simplicity, we describe the Gaussianized version of this model below, though our results also apply analogously to the original SBM version. 

% Deshpande, Montanari, Mossel and Sen \cite{deshpande2018contextual} (see also follow up work) studied the following model for community detection/clustering. 
The generative model is $v \sim \mathsf{Uniform} (\{\pm 1\}^n)$, $u \sim N(0,I_p/p)$, $W$ is a GOE matrix, i.e., a symmetric matrix where independently $W_{ij} \sim N(0,1/n)$ for $i < j$ and $W_{ii} \sim N(0,2/n)$, and $Z \in \R^{ n \times p}$ is a matrix with iid $N(0,1/p)$ entries. Then we observe
\[ A = \frac{\lambda}{n} vv^\top  + W, \qquad B = \sqrt{\frac{\mu}{n}} v u^\top  + Z, \qquad u \sim N(0,I_p/p). \]
Informally, words $A_{ij}$ is some indication of whether $v_i$ and $v_j$ are likely to agree, and rows of $B$ are context/feature vectors in $\mathbb{R}^p$
from a mixture of two spherical gaussians with means $\pm \sqrt{\mu/n}\ u$, where each gaussian corresponds to one community assignment.
In this model, the posterior (see Appendix~\ref{a:example} for the derivation) is
$p(v \mid A,B) \propto \exp\left(\frac{\lambda}{2} \langle vv^\top , A \rangle + \frac{p\mu}{2n(1 + \mu)}\langle vv^\top , BB^\top  \rangle \right)$,
so it is an Ising model where the interaction matrix is the weighted sum of $A$ and $BB^\top $. We can sample from this using our result as long as the dimension $p$ is fixed (since $BB^\top $ is rank at most $p$) and provided $\lambda\|A\|_{\op} < 1/2$. Note that if $A$ is actually generated from the model, then $\|A\|_{\op} \le 2(1 + o_{n \to \infty}(1))$ due to well-known results on spiked Wigner matrices (see \cite{perry2018optimality} and references within)
%by classical results \cite{anderson2010introduction} 
in which case we would have mixing for $\lambda < 1/4$. Like our previous application, the sampler works fine with any context matrix $B$. 

%Although they studied $p$ proportional to $n$, we can consider also $p$ of order $1$ i.e., low-dimensional contexts. In this case, it seems we should actually take $\mu$ on the order of $n$, which is different from what would naively predict from the proportional scaling setting. The reason is that in 1 dimension, we need a constant separation to achieve nontrivial clustering in a gaussian mixture model (whereas the threshold for recovering the planted direction $u$ is smaller). 
%% THIS WAS WRONG, it's a ifferent behavior
%in which case we should take $\mu^2$ on the same order as $n$. NOTE: the trick we use to decompose into mixture of Ising models is exactly the opposite of the integration over $u$ above. 

%TODO: check that when $\mu^2 = \Theta(n)$, having $A$ in addition to $B$ is useful for estimating e.g., $vv^\top $. (Seems reasonably likely?) 

\input{ack}

%\printbibliography
\bibliography{bib}

\appendix

%\section{Appendix}

\newpage

\section*{Overview of Appendix}
The Appendix includes complete proofs of all of the main results. We set out notations and definitions in Appendix~\ref{a:not-df}. 
Appendix~\ref{sec:negdef} formalizes the argument for handling negative outlier eigenvalues. Appendix~\ref{sec:Z} gives the proof of the part of Theorem~\ref{t:main} for estimating the partition function, and Appendix~\ref{s:sampling} gives the proof for sampling. 
Appendix~\ref{app:equiv} provides a re-interpretation of the Hubbard-Stratonovich transform in terms of Gaussian mixture posteriors, and 
%Appendices~\ref{app:equiv} and 
Appendix~\ref{app:technical} contains supporting technical lemmas for the previous sections. Appendix~\ref{a:example} contains additional calculations related to the examples. Finally, we prove the computational hardness results in Appendix~\ref{s:hardness}.  
\section{Notation and definitions}
\label{a:not-df}
\subsection{Notation}
For a set $I\subeq A$, we let $I\cdot c := \set{cx}{x\in I}$; for instance, $\wh Z\in Z \cdot [\rc 2,2]$ means $\rc 2 Z\le \wh Z\le 2Z$.

We will often omit subscripts and superscripts for probability distributions; when we need to be precise, we will indicate the variables as superscripts (for example, $p^{\si,\mu}$, $p^{\si|\mu}$). We use a lowercase letter $p$ to denote the probability density functions and an uppercase letter $P$ to denote the corresponding probability measure.
All probability densities are with respect to the uniform measure on the hypercube and Lebesgue measure on $\R^n$. When we write $\propto$, the constants of proportionality do not depend on the variables to the left of the conditioning. 

\input{notations}

\subsection{Background on Markov chains}

\label{a:def}
Let $P$ be a measure on some space $\Om$ and $T$ be the transition kernel of the ``natural" Markov chain associated with $P$, e.g., Glauber dynamics (Algorithm~\ref{a:mc}) when $P$ is defined on the hypercube $\Om = \{\pm 1\}^n$. 
The Poincar\'e and modified log-Sobolev constants of  $P$ are defined as
\begin{align*}
    \CP(P) &=  \sup\set{\fc{\Var_P(f)}{\cE_P(f,f)}}{f:\BC\to \R, \Var_P(f)\ne 0}\\
    \CMLS(P) &= \sup\set{\fc{2\Ent_P(f)}{\cE_P(f,\log f)}}{f:\BC\to \R_{\ge 0}, \Ent_P(f)\ne 0}
\end{align*}
where $\Ent_P(f) = \E_P[f\log f] - \E_P[f]\log \E_P[f]$, and
\begin{align*}
    \cE_P(f,g) &= \E_P[f \cdot \cL_Pg]\\
    \text{where}\quad \cL_Pf &= (\id - T)f.
\end{align*}
In particular, for Glauber dynamics on $\BC$,
\begin{align*}
\quad (\cL_Pf)(\si) &=\rc n \sumo in \pa{\E_P[f(x)|x_{-i}=\si_{-i}] - f(\si)}.
\end{align*}
Here, for $\si\in \BC$, $\si_{-i}\in \{\pm1\}^{n-1}$ denotes all coordinates except the $i$th one.
Note that some texts use instead the reciprocal of $\CP,\CMLS$, or do not include the $\rc n$.

We also define the Cheeger constant of the Markov chain by 
\begin{align*}
\Phi &= \min_{A\subeq \Om,P(A)\le \rc 2} \fc{Q(A,A^c)}{P(A)}\\
\text{where}\quad Q(A,B) &= \int_{A} T(x,B) \,P(dx).
\end{align*}

\input{negdef}
\input{Z}
\input{sampling}
\input{equiv2}
\input{appendix}

\input{hardness}
\end{document}

%% file: packages_arxiv.tex
\usepackage{etex}
\usepackage{fullpage}

\usepackage{amsmath}
\usepackage{amssymb}
\usepackage{amsthm}
\usepackage{array}
\usepackage{bbm}
\usepackage{cancel}
\usepackage{cmap}
\usepackage{enumerate}
\usepackage{enumitem}%resume lists
\usepackage{fancyhdr}
\usepackage{mathdots}%iddots: dots going northeast
\usepackage{mathtools}
\usepackage{mathrsfs}
\usepackage[colorlinks=true, citecolor=blue]{hyperref}
\usepackage{stackrel}
\usepackage{stmaryrd}%\mapsfrom
\usepackage{tabularx}
\usepackage{tikz}
\usepackage{ctable}
\usepackage{titlesec}
\usepackage{titletoc}
\usepackage{url}
\usepackage{verbatim}
\usepackage{wasysym}
\usepackage{wrapfig}
\usepackage{yhmath}
\usepackage[all,cmtip]{xy}%Commutative diagrams
\usepackage{hyperref}
\usepackage{natbib}
\usetikzlibrary{calc,trees,positioning,arrows,chains,shapes.geometric,%
    decorations.pathreplacing,decorations.pathmorphing,shapes,%
    matrix,shapes.symbols,shadows,fadings}

%% file: theorems_std.tex
\newtheorem{thm}{Theorem}[section]
\newtheorem*{thm*}{Theorem}

\newtheorem*{prb*}{Problem}

\newtheorem*{ax*}{Axiom}

\newtheorem*{clm*}{Claim}

\newtheorem*{conj*}{Conjecture}
\newtheorem{cor}[thm]{Corollary}
\newtheorem{df}[thm]{Definition}
\newtheorem*{df*}{Definition}

\newtheorem*{ex*}{Example}

\newtheorem{lem}[thm]{Lemma}
\newtheorem*{lem*}{Lemma}

\newtheorem*{pos*}{Postulate}
\newtheorem{pr}[thm]{Proposition}
\newtheorem*{pr*}{Proposition}

\newtheorem*{qu*}{Question}
\newtheorem{rem}[thm]{Remark}
\newtheorem*{rem*}{Remark}

%% file: macros.tex
\newcommand{\cE}[0]{\mathcal{E}}%deprecated

\newcommand{\E}[0]{\mathbb{E}}

\newcommand{\cL}[0]{\mathcal{L}}

\newcommand{\Pj}[0]{\mathbb{P}}

\newcommand{\R}[0]{\mathbb{R}}

\newcommand{\one}[0]{\mathbbm{1}}
\newcommand{\be}[0]{\beta}
\newcommand{\ga}[0]{\gamma}

\newcommand{\de}[0]{\delta}

\newcommand{\ep}[0]{\varepsilon}
\newcommand{\eph}[0]{\frac{\varepsilon}{2}}

\newcommand{\la}[0]{\lambda}

\newcommand{\Te}[0]{\Theta}

\newcommand{\om}[0]{\omega}
\newcommand{\Om}[0]{\Omega}
\newcommand{\si}[0]{\sigma}

\newcommand{\nin}[0]{\not\in}

\newcommand{\subeq}[0]{\subseteq}

\newcommand{\iy}[0]{\infty}
\newcommand{\rc}[1]{\frac{1}{#1}}
\newcommand{\prc}[1]{\pa{\rc{#1}}}

\newcommand{\fc}[2]{\frac{#1}{#2}}
\newcommand{\sfc}[2]{\sqrt{\frac{#1}{#2}}}
\newcommand{\pf}[2]{\pa{\frac{#1}{#2}}}%Shortcut for fraction with parentheses

\newcommand{\dd}[2]{\frac{d #1}{d #2}}

\newcommand{\nb}[0]{\nabla}

\newcommand{\ab}[1]{\left| {#1} \right|}
\newcommand{\an}[1]{\left\langle {#1}\right\rangle}
\newcommand{\ba}[1]{\left[ {#1} \right]}
\newcommand{\bc}[1]{\left\{ {#1} \right\}}

\newcommand{\pa}[1]{\left( {#1} \right)}

\newcommand{\ve}[1]{\left\Vert {#1}\right\Vert}

\newcommand{\set}[2]{\left\{{#1}:{#2}\right\}}

\newcommand{\ol}[1]{\overline{#1}}

\newcommand{\wt}[1]{\widetilde{#1}}
\newcommand{\wh}[1]{\widehat{#1}}

\newcommand{\Ent}{\operatorname{Ent}}

\newcommand{\Gap}{\operatorname{Gap}}

\newcommand{\id}{\mathrm{id}}

\newcommand{\poly}{\operatorname{poly}}
\newcommand{\polylog}{\operatorname{polylog}}

\newcommand{\rank}{\operatorname{rank}}

\newcommand{\Tr}[0]{\operatorname{Tr}}

\newcommand{\TV}[0]{\textup{TV}}
\newcommand{\val}[0]{\text{val}}
\newcommand{\Var}[0]{\operatorname{Var}}

\providecommand{\cal}[1]{\mathcal{#1}}
\renewcommand{\cal}[1]{\mathcal{#1}}

\newcommand{\pull}[9]{
#1\ar@/_/[ddr]_{#2} \ar@{.>}[rd]^{#3} \ar@/^/[rrd]^{#4} & &\\
& #5\ar[r]^{#6}\ar[d]^{#8} &#7\ar[d]^{#9} \\}

\newcommand{\cmp}[9]{
\xymatrix{
#1 \ar[r]^{#4}{#5} \ar@/_2pc/[rr]^{#8}_{#9} & #2 \ar[r]^{#6}_{#7} & #3
}
}

\newcommand{\ha}[1]{\ar@{^(->}[#1]}
\newcommand{\ls}[1]{\ar@{-}[#1]}
\newcommand{\sj}[1]{\ar@{->>}[#1]}
\newcommand{\aq}[1]{\ar@{=}[#1]}
\newcommand{\acir}[1]{\ar@{}[#1]|-{\textstyle{\circlearrowright}}}
\newcommand{\acil}[1]{\ar@{}[#1]|-{\textstyle{\circlearrowleft}}}
\newcommand{\ard}[1]{\ar@{.>}[#1]}
\newcommand{\mt}[1]{\ar@{|->}[#1]}
\newcommand{\inm}[1]{\ar@{}[#1]|-{\in}}
\newcommand{\inr}{\ar@{}[d]|-{\rotatebox[origin=c]{-90}{$\in$}}}
\newcommand{\inl}{\ar@{}[u]|-{\rotatebox[origin=c]{90}{$\in$}}}

\newcommand{\sumo}[2]{\sum_{#1=1}^{#2}}

\newcommand{\prodo}[2]{\prod_{#1=1}^{#2}}

\newcommand{\cupo}[2]{\bigcup_{#1=1}^{#2}}

\newcommand{\beq}[1]{\begin{equation}\llabel{#1}}
\newcommand{\eeq}[0]{\end{equation}}
\newcommand{\bal}[0]{\begin{align*}}
\newcommand{\eal}[0]{\end{align*}}%this doesn't work; i don't know why
\newcommand{\ban}[0]{\begin{align}}
\newcommand{\ean}[0]{\end{align}}

\newcommand{\fixme}[1]{{\color{red}#1}}
\newcommand{\llabel}[1]{\label{#1}\text{\fixme{\tiny#1}}}

\newcommand{\arxiv}[1]{\url{http://www.arxiv.org/abs/#1}}

\newcommand{\vocab}[1]{\textbf{#1}} %also index automatically.

\allowdisplaybreaks[2]

\DeclareFontFamily{U}{wncy}{}
    \DeclareFontShape{U}{wncy}{m}{n}{<->wncyr10}{}
    \DeclareSymbolFont{mcy}{U}{wncy}{m}{n}
    \DeclareMathSymbol{\Sh}{\mathord}{mcy}{"58} 

%% file: intro.tex
\section{Introduction}
An Ising model is a probability distribution on the hypercube $\{\pm 1\}^n$ of the form
\[ p_{J,h}(\si) = \frac{1}{Z} \exp\left(\rc2 \an{\si, J\si} + \an{h, \si}\right) \]
where the normalizing constant $Z$ is known as the \emph{partition function}. 
The closely related problems of estimating the partition function $Z$ and sampling from the Ising model are
%Estimating the partition function of an Ising model 
fundamental computational problems, both due to their central theoretical significance as well a plethora of applications---see for example \cite{mezard2009information,talagrand2010mean,wainwright2008graphical,jerrum1996markov,hinton2012practical,murphy2012machine}.
While computing the partition function $Z$ exactly is $\#\mathsf{P}$-hard \citep{jerrum1993polynomial}, and approximating it is $\mathsf{NP}$-hard (see e.g., \cite{sly2012computational,galanis2016inapproximability}), a vast amount of work has been done to understand and characterize situations where this task is computationally tractable. 
%While the general problem is \#P hard to even approximate \cite{}, much work has been done on finding special tractable instances.  

One of the dominant approaches in both theory and practice to sample from such models is the \emph{Glauber dynamics} or \emph{Gibbs sampler}. This is a Markov chain that at each step, resamples the spin of one coordinate from its conditional distribution. In general, this chain is expected to mix under appropriate assumptions on the weakness of the interactions in the model (e.g., presence of correlation decay, or uniqueness of the corresponding Gibbs measure on the tree). In certain special cases, the point at which the Glauber dynamics stops mixing rapidly is also exactly where sampling becomes hard: famously, this is the case for the antiferromagnetic Ising model on the worst-case $d$-regular graph (see e.g., \cite{sly2012computational,chen2020rapid}). However, this is not the case in general---there are many examples where Glauber dynamics fails to mix but other methods succeed to approximate the partition function and/or sample; see e.g., \cite{jerrum1993polynomial,borgs2020efficient,risteski2016calculate,guo2017random} for a few examples. 
%\emph{high temperature} (weak interaction) assumptions,  

\emph{Variational methods} are %probably 
the main alternative to MCMC (Markov Chain Monte Carlo) methods in practice. In general, variational methods attempt to reduce to problem of computing the partition function to solving an optimization problem---see e.g., \cite{wainwright2008graphical,mezard2009information} for further background. Importantly, the strengths and limitations of variational methods are complementary to those of Glauber dynamics. Unlike Markov chain methods, variational methods are usually based on solving for an \emph{approximation} of the true distribution, and hence may only achieve a comparatively crude approximation to the true distribution---a successful variational approximation may only output a distribution with KL divergence or Wasserstein distance $o(n)$ as opposed to $o(1)$ for the output of a rapidly mixing Markov chain. On the other hand, variational methods often work in both high and low-temperature settings and are closely related to textbook methods for solving low-temperature models, such as the Ising model on a high-dimensional lattice, the Curie-Weiss model, and the Sherrington-Kirkpatrick model \citep{talagrand2010mean,mezard2009information,parisi1988statistical}.
%Particularly relevant to the present work is a major alternative to Markov chain methods known as \emph{variational inference}. 
%methods} or variational inference. 

To give a concrete example with strong theoretical guarantees,  the \emph{naive mean-field approximation}, which corresponds to approximating the Gibbs measure by a (small mixture of) product measure(s), is probably the most well-known variational method. It has been established that this approximation is in various senses accurate whenever the interaction matrix $J$ has quantitatively low rank (more precisely, when $\|J\|_F^2 = \sum_i \lambda_i(J)^2 = o(n)$): see \cite{basak2017universality,eldan2018gaussian,eldan2018decomposition,eldan2020taming,augeri2021transportation} for a few of the works in this area. This condition essentially covers all of the main examples of Ising models where the mean-field approximation is known to be accurate, and for these models it covers both low and high temperature regimes
(i.e., both strong and weak couplings). %HL?
Correspondingly, there are approximation algorithms connected with the naive mean-field approximation \citep{risteski2016calculate,jain2018mean,jain2018vertex,jain2019mean} which approximate $\log Z$ within $o(n)$ additive error in subexponential time under this assumption (with improving runtime as the rank decreases, and with roughly matching computational lower bounds). 

In this work, we seek to achieve the \emph{best of both worlds} and combine the strengths of Glauber dynamics and variational inference. Recently, it was shown \citep{eldan2020spectral,anari2021aentropic} that the Glauber dynamics rapidly mix whenever the eigenvalues of $J$ all lie  within an interval of length $1$, which is tight due to the example of the Curie-Weiss model \citep{levin2017markov}. Our main result shows that by using a more sophisticated algorithm, we can sample in polynomial time from any Ising model with a constant number of eigenvalues outside of this interval, a situation which occurs in many examples of interest. To state our result, first note that without loss of generality, we can recenter the bulk of the eigenvalues to $[0,1]$ by adding a multiple of the identity to $J$. We provide an algorithm that samples from an Ising distribution with $d_+$ eigenvalues bigger than $1-1/c, c \in (1,\infty]$, and $d_-$  negative eigenvalues $-\lambda_1, \dots, -\lambda_{d_-}$ in time $(n\ve{J}_\op)^{O(d_+)}e^{O(c(\lambda_1 + \dots + \lambda_{d-}))}$, as well as (multiplicatively) approximate the partition function.
%of the model. 

In the special case of low-rank Ising models where the naive mean-field approximation is accurate, this gives a roughly comparable runtime to the previous approximation algorithms for estimating $\log Z$ (e.g., \cite{jain2019mean}), while allowing us both to approximate $Z$ much more accurately (within an arbitrary multiplicative factor) and also to sample; see Remark~\ref{rem:speed-mf} for further discussion.  Our result also allows us to sample from models which are genuinely high-rank, for example the SK model with ferromagnetic interactions in the regime where the bulk has diameter at most $1$ (see Section~\ref{sec:applications}) in which case the naive mean-field approximation is known to be very inaccurate (see e.g., \cite{thouless1977solution,jain2019mean}). Our general result also continues a long tradition of seeking fixed-parameter tractable algorithms for optimization problems that are ``approximately'' low rank \citep{frieze1996regularity, oveis2013new}.
Our techniques take inspiration from both variational and MCMC approaches. 
We describe them in detail later (see Section~\ref{sec:technical-overview}),
%\fnote{later, add a reference to some kind of technical overview?}
but at a high-level our result is based on two key innovations: (1) for positive outlier eigenvalues, a rigorous version of the popular \emph{simulated annealing}~\citep{lovasz2006simulated} and \emph{tempering} heuristics~\citep{marinari1992simulated}, based in part on a decomposition of the measure into a mixture of high-temperature Ising models using the Hubbard-Stratonovich transform \citep{hubbard1959calculation}, and (2) for negative eigenvalues, a sampling approach based on importance sampling combined with the efficient solution of a related fixed point equation, which is done by constructing an appropriate (nonconvex) variational problem and running stochastic gradient descent.
%a new variational problem and SGD to find a critical point of this optimization.
%for which we constructively show existence of a critical point by running stochastic gradient ascent on the nonconvex objective. 
The key ideas behind both steps are clean and we believe the techniques may be useful for solving other sampling problems of interest.
%\anote{Also say something on decomposition? how we reduce to 1 and 2. Prob also say something more about the variational problem and why we can solve a non-convex prob. } 

%The class of models we consider is also natural in its own right---in combinatorial optimization, there is a long tradition of fixed-parameter tractable algorithms for instances that are (close-to) low rank \citep{frieze1996regularity, oveis2013new}. 
In addition to this, we provide representative applications of our results to a diverse set of tasks: %We show:
%As a few representative applications, we show: 
%itemize leads to ugly + huge whitespace...
%\begin{itemize}
%    \item 
First, we give an algorithm to sample Ising models (antiferromagnetic or ferromagnetic, and potentially with inconsistent external fields) on expander graphs up to inverse temperature $\beta = O(1/\lambda)$ where $\lambda$ is the second largest eigenvalue. This is outside the tree uniqueness regime; note that on general graphs, antiferromagnetic Ising is NP-hard past this threshold \citep{sly2012computational}. Also, even when the model is ferromagnetic, inconsistent external fields make the sampling problem \#BIS-hard in general\footnote{Our results work in an expanded ``high temperature'' regime; in contrast algorithms for different \#BIS-hard problems work in a low temperature regime by expanding around the ground states \citep{jenssen2020algorithms, chen2021fast}, so these approaches should be naturally complementary when they both apply.}.
    %In the antiferromagnetic Ising model, ground states can be computationally hard to compute, which would prevent sampling at low temperatures.}. 
    Relatedly, we give the first results for sampling high-temperature Sherrington-Kirkpatrick models with strong ferromagnetic interactions.
    %Our results work in a  complementary to recent works tackling the low-temperature regime, by techniques involving Taylor expanding the partition function. \citep{jenssen2020algorithms, chen2021fast} \fnote{They also study a problem which is \#BIS-hard in general (ferromagnetic potts), but there is no actual overlap with our applications right? Except Ising without external field, which is solved by Jerrum-Sinclair anyway.}
    %\anote{How does this relate to Perkins?} \fnote{is it this paper? \url{https://arxiv.org/pdf/1807.04804.pdf} I think their technique is mostly for problems where the optimization version is very easy to solve, like ferromagnetic Potts model (their main application). so afaik it doesn't generally apply to the problems we consider though we should cite it} \hlnote{See also~\cite{chen2021fast}. What does the formal reduction between \#BIS and ferromagnetic Ising give?}
    %\item 

We also show how to sample from a Hopfield network \citep{hopfield1982neural} with a fixed number of patterns in polynomial time. As an example Bayesian statistics application, we show how to sample from posteriors of mixtures of two Gaussians with symmetric means in fixed dimension. This provides complementary results to \citep{mou2019sampling}, who consider the same setting in an arbitrary dimension, but instead consider an easier task: sampling from the so-called power posterior of such a mixture---which is derived by weighing the prior substantially more in the Bayes formula for the posterior. More generally, we show how to sample from a regime of a more sophisticated clustering model (the Contextual Stochastic Block Model) with low-dimensional contexts.
    %\fnote{they also require the data to be well-specified or drawn from a slightly Huber contaminated model}
    %\anote{How does this relate to Mou et al? } \fnote{is that the SBM one? replied on slack} \hlnote{\cite{mou2019sampling} They consider arbitrary dimension, but ``weaken" the problem by considering a power posterior, which they sample from on the $\mu$ space; it's actually not amenable to sampling with Glauber dynamics}
    %\item  
%    \item We show for the first time how to sample from a Hopfield network \citep{hopfield1982neural} with a fixed number of patterns in polynomial time.
%\end{itemize}
%We complement our algorithm with a negative result, showing that $J$ is low-rank but has extremely large eigenvalues (which as discussed earlier, slow down our algorithm)
%We complement our algorithmic results with some negative results...

\subsection{Main results}
Suppose that $J$ is %(symmetric) positive semidefinite. 
a symmetric matrix. 
We are interested in and computing the partition function $Z_{J,h}$ and sampling from the distribution $P_{J,h}$ over $\{\pm 1\}^n$ given by
\begin{align}\label{e:ising}
p_{J,h}(\si) &= \fc{\eJhs}{Z_{J,h}}, & 
\text{where }Z_{J,h} &= \sum_{\si\in \{\pm 1\}^n} \eJhs.
\end{align}
Our main theorem is the following.
\begin{thm}\label{t:main}
Let $c\in (1,\iy]$, $\ep\in (0,1)$. 
Suppose that $J$ is a symmetric matrix such that (1) $J$ has $d_+$ eigenvalues that are greater than $1-\rc c$, and (2) its negative eigenvalues are $-\la_1,\ldots, -\la_{d_-}$.
\begin{enumerate}
     \item There is an algorithm (Algorithm~\ref{a:Z-ising}) that with probability $\ge 1-e^{-n}$, gives a $e^\ep$-multiplicative approximation to $\ZJh$ in time $O\pa{(\ve{J}_\op n)^{O(d_++1)} e^{O(c(\la_1+\cdots + \la_{d_-}))}\big/\ep^2}$.
    %$O\pa{\sqrt n \log\pf n\ep}^{d_+} \cdot O\pf{\poly(n)}{\ep^2} \cdot e^{O(c(\la_1+\cdots + \la_{d_-}))}$.
    %\hlnote{Can write more simply as $n^{O(d_++1)} e^{O(c(\la_1+\cdots + \la_{d_-}))}/\ep^2$, as $\log(1/\ep)$ only matters when $\ep$ is exponential, when we can brute-force}
    \item There is an algorithm (Algorithm~\ref{a:st}) to sample from a distribution $\ep$-close in TV-distance to $P_{J,h}$ 
 in time $\pa{\ve{J}_\op n \log \prc \ep}^{O(1+d_+)} e^{O(c(\la_1+\cdots + \la_{d_-}))}$.
\end{enumerate}
\end{thm}
Note that we can take $c=\iy$ in the theorem; in this case we assume that $J$ has no negative eigenvalues, i.e., $J$ is positive semi-definite, and we get the simpler bounds $O\pa{(\ve{J}_{\op} n)^{O(d_++1)}\big/\ep^2}$ and $\pa{\ve{J}_\op n \log \prc \ep}^{O(1+d_+)}$. 
%\hlnote{Also need to give bound separately for $c=0$ (do that case first?), as this bound blows up.}
%\fixme{Negative result here.}
Excluding the dependence on $\ve{J}_\op$, for large positive eigenvalues the runtime only depends on the number of eigenvalues, but for negative eigenvalues, the runtime depends on their magnitude. %\fnote{is this entirely true, or just the quantitative dependence is different}

When there are $n$ large eigenvalues, our runtime guarantee is similar to brute force\footnote{Note however, that Theorem~\ref{t:main} only gives nontrivial guarantees when $d_+=o\pf{n}{\log n}$; it is an interesting question whether one can remove the $\log n$ factor.}; see \citep{jain2019mean} for discussion of why this should be unavoidable under the Exponential Time Hypothesis (ETH).
In the extreme case where there is just a single very large negative eigenvalue, it turns out the problem is also computationally hard. This arises from the discrete nature of the hypercube $\{ \pm 1\}^n$ and stands in strong contrast to intuition from sampling continuous distributions, where very strong log-concavity is not an obstacle to efficient sampling.
We prove the following negative result; see the full theorem (Theorem~\ref{thm:hardness}) for a stronger runtime lower bound for estimating $\log Z$, conditional on the ETH.
%The full theorem also gives a very strong runtime lower bound for estimating $\log Z$, conditional on the ETH.
\begin{thm}[Theorem~\ref{thm:hardness}]\label{thm:hardness-intro}
Let $\beta \ge 1$ be arbitrary and fixed.
For any $a = (a_1,\ldots,a_n) \in \mathbb{Z}^n$, define the Ising model with probability mass function $p_a : \{\pm 1\}^n \to [0,1]$ given by
$p_a(\sigma) %= \frac{1}{Z}
\propto \exp\left(-\beta n \langle a, \sigma \rangle^2\right)$.
%_{4\be naa^\top}
%where the interaction matrix $J$ is rank one and negative definite, then $NP = RP$.
If there exists a polynomial time randomized algorithm to approximately sample within TV distance $1/2$ from Ising models of this form for any $a_1,\ldots,a_n$,
then $\mathsf{NP} = \mathsf{RP}$.
\end{thm}

%% file: overview.tex
\section{Overview of techniques}\label{sec:technical-overview}
This section has two parts: in the first, we recall some basic tools which we will use in our analysis. In the second, we give a full overview of our algorithm and the proof of our main result. 
\subsection{Technical toolkit}
\paragraph{Sampling from Ising models with bounded spectral diameter.}
As a basic ingredient, we use the following guarantee for Glauber dynamics on Ising models (see also \cite{bauerschmidt2019very,eldan2020spectral}):
\begin{thm}[{\cite[Theorem 12]{anari2021aentropic}}]\label{thm:glauber}
Let $J\in \R^{n\times n}$ be a symmetric matrix satisfying $0 \preceq J \prec I_n$, $h \in \mathbb{R}^n$ arbitrary. 
%, and define the Ising model on the hypercube $\{\pm 1\}^n$ by
%\[ p_{J,h}(\sigma) \propto \exp\left(\frac{1}{2} \langle \sigma, J \sigma \rangle + \langle h, \sigma \rangle\right). \]
Then we have that:
\begin{enumerate}
    \item The Poincar\'e and modified Log-Sobolev constants %of the discrete-time Glauber dynamics on $p_{J,h}$ is at least $(1 - \|J\|_\op )/n$.
    of $P_{J,h}$ are at most $n(1-\ve{J}_\op)^{-1}$. 
    \item For any $\epsilon > 0$, the discrete-time Glauber dynamics mixes to $\epsilon$ total variation distance of $P_{J,h}$ in $O(n\log(n/\epsilon)/(1 - \|J\|_\op ))$ steps.
\end{enumerate}
\end{thm}
See Appendix~\ref{a:def} for the definition of the Poincar\'e and modified log-Sobolev constant. 
%\fnote{We might not actually use modified log-sobolev constant, I guess we use Poincar\'e constant somewhere? check and define} \hlnote{I used modified log-Sobolev and Poincar\'e constant as the inverse of what's given here, also I was referencing the normalization without the $1/n$.} \fnote{maybe change the statement here and include the definition of the notation your way?}
%\hlnote{Move Lemma~\ref{l:ls-ising} here?}

%\hlnote{Is there a reason it's ``diameter" instead of ``radius"?} \fnote{mostly because radius $1/2$ seems less natural than diameter $1$ i guess}

\paragraph{Hubbard-Stratonovich transform.}
%\paragraph{Hubbard-Stratonovich transform: an equivalent sampling problem}
The component of our algorithm which handles positive spike eigenvalues makes use of the multivariate version of the classical Hubbard-Stratonovich transform \citep{hubbard1959calculation}. This transform is commonly used in the analysis of quantum and statistical physics systems and in large deviation theory; for a few examples see \citep{talagrand2010mean,bovier2012mathematical,bauerschmidt2019very,hsu2012tail}. 
The statement is given by Lemma~\ref{l:hs} below; it is very useful despite its simplicity. 
\begin{lem}\label{l:hs}
Let $X\in \R^{m\times n}$ be a matrix with $d$-dimensional column space $V$. Let $\si\in \R^n$. Then for any $\gamma > 0$,
\begin{align*}
    \exp\pa{\frac{\gamma^2}{2} \ve{X\si}^2}
    & = \pf{1}{2\pi\gamma^2}^{d/2} \int_V \exp\pa{\an{X^\top \mu, \si} - \frac{1}{2\gamma^2} \ve{\mu}^2}\,d\mu.
\end{align*}
\end{lem}
\begin{proof}%\hlnote{Some missing $n$ factors here} %fixed?
We complete the square to find that
\begin{align*}
    &\pf{1}{2\pi \gamma^2}^{d/2} \int_V \exp\pa{\an{X^\top \mu, \si} - \fc{1}{2\gamma^2} \ve{\mu}^2}\,d\mu\\
    &= \pf{1}{2\pi\gamma^2}^{d/2} \exp\pa{\frac{\gamma^2}{2} \ve{X\si}^2}
    \int_V \exp\pa{-\frac{1}{2\gamma^2} \ve{\mu - X\si}^2}\,d\mu =  \exp\pa{\frac{\gamma^2}{2} \ve{X\si}^2}
\end{align*}
using the formula for the normalizing constant of a Gaussian distribution.
\end{proof}

\subsection{Proof overview}
The proof of our main result, Theorem~\ref{t:main}, combines two modular algorithmic ideas: a grid partitioning and simulated annealing/tempering strategy which handles the large positive eigenvalues, and an optimization and rejection sampling based strategy which handles the negative ones. %We describe each of these strategies separately, because they are fairly easy to combine. \fnote{maybe say a tiny bit more here}
%\fnote{maybe add a very short section at end to explain how they can be combined?}

We briefly comment on the relation between our techniques and those used in the aforementioned literature on naive mean-field approximation, which do not seem as useful for sampling. In all of those works (algorithmic or non-algorithmic), the primary goal is to estimate $\log Z$ within an additive error which is small compared to $n$, but essentially always $\omega(1)$ as $n \to \infty$. The main reason for this is that the naive mean-field approximation is simply not accurate to $O(1)$ additive error even in relatively basic examples (see e.g., \cite{eldan2020taming}). On the other hand, in almost all of those works (and also for Dense Max-CSP, e.g. \cite{frieze1996regularity}) the techniques used are general as far as the form of the distribution concerned: e.g., they can handle a log-likelihood which is not a quadratic function but a higher-order polynomial. Our analysis is based on decomposing the spectrum of the interaction matrix, which only seems to makes sense in the Ising case. 

\subsubsection{Large positive eigenvalues: decomposition and simulated tempering} 
Here we describe our method for sampling from Ising models with large positive eigenvalues. For simplicity, we describe the algorithm when the interaction matrix $J$ is positive semidefinite and return to the general case later. %Once we derive an algorithm for the negative eigenvalue case, it can be plugged in directly inside of the final simulated tempering chain. 
%Ultimately, we need to tweak the algorithm to combine it with the method for outlier negative eigenvalues (next section) but these changes are relatively minor. \fnote{maybe say a tiny bit more somewhere}\anote{this seems a bit deflationary. maybe say it a bit less anticlimatically.}

\paragraph{Warmup: Curie-Weiss model and generalizations.} To motivate our approach, we start with a special case: sampling from a rank-one Ising model of the form
$p_{ww^\top,0}(\sigma) \propto e^{\langle w, \sigma \rangle^2/2}$.
This means the interaction matrix is simply $ww^\top $. A classical example of such a distribution is the \emph{Curie-Weiss} model, in which case $w = \beta \vec{1}/\sqrt{n}$ where $\beta \ge 0$ is referred to as the inverse temperature. It is well known \citep{ellis2006entropy,talagrand2010mean} that the Curie-Weiss model exhibits symmetry breaking in its low temperature phase $\beta > 1$:  the distribution becomes close to supported on two clusters of points, one with $\frac{1}{n} \sum_i \sigma_i \approx y$ and an opposite one with $\frac{1}{n} \sum_i \sigma_i \approx -y$ where $y$ is a nontrivial (i.e., nonzero) solution of the fixed point equation $y = \tanh(\beta y)$. Because Glauber dynamics becomes trapped in one of the clusters, it will not mix \citep{levin2017markov}. 

There are many alternative algorithms to sample from the Curie-Weiss model. For example, the random variable $\sum_i \sigma_i$ is an integer between $-n$ and $n$ and it is straightforward to write down its distribution under the Curie-Weiss model explicitly, letting us sample it; this can also be used with a Markov chain decomposition theorem to show mixing up to phase~\citep{madras2003swapping}.
%(which makes sampling $\sigma$ easy). %By symmetry, sampling $\sigma$ given $\sum_i \sigma_i$ is also easy.
\emph{However, } this approach which works well for the Curie-Weiss model does not generalize nicely --- for a typical vector $w$, $\langle w, \sigma \rangle$ will take on $2^n$ many different values!
%and so sampling this random variable is no longer trivial. 
There are multiple ways to provably sample from ferromagnetic Ising models which apply to Curie-Weiss \citep{jerrum1993polynomial,guo2017random}, but we need to also sample from non-ferromagnetic ones.
%with consistent external field (models where the interaction matrix and $h$ are nonnegative, see \cite{jerrum1993polynomial,guo2017random}), but general rank one models are not ferromagnetic.

We now explain an approach that \emph{will} generalize nicely to rank-one models and beyond. We first describe this as a method to compute the partition function $Z$, and explain sampling at the end of this section. By applying the Hubbard-Stratonovich transform (Lemma~\ref{l:hs}), we have
\begin{align*} 
Z = \sum_{\sigma \in \{\pm 1\}^n} e^{\langle w, \sigma \rangle^2/2}
%&= \sum_{\sigma \in \{ \pm 1\}^n} \int_y e^{y\langle w, \sigma \rangle - y^2/2} dy \\
&= \int_{\R^n} e^{-y^2/2} \sum_{\sigma \in \{\pm 1\}^n} e^{y\langle w, \sigma \rangle} dy
= 2^n \int_{\R^n} e^{-y^2/2} \prod_{i = 1}^n \cosh(y w_i)\ dy.
\end{align*}
This is a one-dimensional integral: it's over an infinite domain, but the term $e^{-y^2/2}$ ensures that larger values of $y$ contribute only a negligible amount to the integral. Hence, we only need to perform an integral over a bounded region which can be done using Riemann summation.

\paragraph{The general case: decomposition and integration.} We now consider the much more general case of a positive semidefinite matrix $J$. We do not want to restrict ourselves to low-rank $J$, but rather $J$ which have a smaller large number of eigenvalues greater than $1$. For this reason, we only apply the Hubbard-Stratonovich transform over the large eigenspaces of $J$.

To do this, let $c > 0$ be an arbitrary small constant. Using the spectral decomposition of $J$, we can decompose $J = J^{\perp} + J^{\prl}$ so that $J^{\perp}$ and $J^{\prl}$ are both positive semidefinite, $\|J^{\perp}\|_{\op} \le 1 - c$, and $J^{\prl}$ spans the eigenspaces of $J$ above $1 - c$, which we denote as $V^{\prl}$ with dimension $d$. Let $J^{\prl} = X^\top  X$ be an arbitrary factorization; then by an analogous application of the Hubbard-Stratonovich transform (Lemma~\ref{l:hs}) we have
\begin{align} 
Z 
&= \sum_{\si\in \BC} \exp\pa{\isingexp{J^\perp}h} \exp\pa{\rc{2} \ve{X \si}^2} \nonumber \\
&= \left(\frac{1}{2\pi}\right)^{d/2} \sum_{\si\in \BC} \exp\pa{\isingexp{J^\perp}h} \int_{V^\prl} \exp\pa{\an{X^\top \mu^\prl , \si} - \frac{1}{2} \ve{\mu^\prl}^2}\,d\mu^\prl \nonumber \\
&= \left(\frac{1}{2\pi}\right)^{d/2} \int_{V^\prl} \exp\left(-\frac{1}{2} \ve{\mu^\prl}^2\right)  \sum_{\si\in \BC} \exp\pa{\isingexp{J^\perp}{h + X^\top \mu^\prl}}\,d\mu^\prl. \label{eqn:hs-apply-intro}
\end{align}
We see the resulting integral is now over a $d$-dimensional subspace; just like the example, the integrand has a damping term $\exp\left(-\frac{1}{2} \ve{\mu^\prl}^2\right)$ which allows us to truncate it to a bounded domain while changing the integral by only a small amount. Each of the integrands involves a sum over exponentially many $\sigma \in \{\pm 1\}^n$, but we can recognize this sum as the partition function of an Ising model with interaction matrix $J^{\perp}$. Since $J^{\perp}$ has no large eigenvalues, and we can sample from this class of models using Glauber dynamics (Theorem~\ref{thm:glauber}), we can approximate the corresponding partition function using a relatively standard reduction from sampling to integration (see e.g., \cite{bezakova2008accelerating}; this reduction is via a form of simulated annealing, not to be confused with the related but different concept of simulated tempering described later). Finally, using Riemann summation to actually compute the integral gives the estimate of $Z$.

\paragraph{The simulated tempering chain: sampling with exponentially small error.} In principle, given the previous result for approximating the partition function, we could apply standard reductions from approximate counting to sampling in order to approximately sample from the Ising model. This would be quite suboptimal, because the running time of such an algorithm would depend polynomially on the error parameter $\epsilon$ (desired total variation distance to the true distribution). %In other words, it would only converge only slowly to the target distribution. 
In comparison, MCMC methods, when they work, %almost always 
generally depend logarithmically on the error parameter $\epsilon$ and we would like our algorithm to have this property too.

To achieve the desired logarithmic dependence on $1/\ep$, we construct a new Markov chain. The first step is to observe that the formula \eqref{eqn:hs-apply-intro} we derived comes with a simple probabilistic interpretation: it can be understood as a decomposition of the original Ising model into a %``small''\anote{how is it a small mixture?} 
mixture of high-temperature Ising models with additional external field $X^\top  \mu^{\prl}$. The associated joint distribution over the pair $(\sigma, \mu^{\prl})$ is
\begin{align}\label{e:psimuprl}
p(\si,\mu^\prl) \propto 
\exp\pa{\isingexp{J^\perp}{h+X^\top \mu^\prl} - \fc 12 \ve{\mu^\prl}^2}.
\end{align}
With this understanding, all we need to do is construct a Markov chain which can sample quickly on the joint $(\sigma, \mu^{\prl})$ space. However, a standard Metropolis-Hastings sampler has the same issue as the original Glauber dynamics: the joint distribution in $(\sigma, \mu^{\prl})$ space is multimodal just like the original distribution. 

The key to solving this problem is to use a faster chain based on \emph{simulated tempering} \citep{marinari1992simulated}. 
%(The name refers to a physical analogy with metallurgy.)  
We actually define the Markov chain on a further expanded state space of $(\ell, \sigma, \mu^{\prl})$ where $\ell$ is an additional \emph{temperature} variable, so that the chain mixes to a distribution which conditional on the temperature $\ell$ being at its ``coldest'' setting is the desired distribution. %The point is that when the chain is moving around at ``hotter'' temperatures, the stationary distribution can be different from the target: better connected and easier to move around in.
%The temperature schedule is chosen so that distributions at adjacent distributions have constant overlap.
The point is that the chain mixes rapidly at the ``hottest" temperature, which combined with 
a choice of temperature schedule where distributions at adjacent distributions have constant overlap, 
provides a bridge between the different modes at the colder temperature. 
We actually consider a variant of simulated tempering where we approximately equalize the probability for each grid cell so that they will all be visited ---this can be thought of as a Markov chain analogue of grid search %in optimization%~\citep{oveis2013new}
--- with a final step of importance sampling to attain the right probabilities.
 % reference doesn't seem necessary? %HL: wanted to point out this is actually very similar to something that has already been done in the low-rank case for optimization

Simulated tempering is a beautiful idea, but it isn't always guaranteed to  work: indeed, \cite{marinari1992simulated} proposed their original simulated tempering chain exactly for the purpose of sampling from Ising models, but it does not come with a mixing time guarantee (and obviously, no sampling method will work for Ising models which are computationally hard to sample \citep{sly2012computational}). In our setting, %and with our variant of the simulated tempering chain, 
we can establish a Poincar\'e inequality and prove rapid mixing by using a \emph{Markov chain decomposition theorem}~\citep{madras2002markov,ge2018simulated}. Such a decomposition theorem allows us to conclude fast mixing once we show mixing within each grid cell as well as a ``coarse-grained" chain where each grid cell is considered as a single state. 
Mixing within each grid cell is immediate from the fact that for fixed $\mu^\prl$, Glauber dynamics for $p(\si,\mu^\prl)$ mixes rapidly, and mixing of the coarse-grained chain follows from equalization of the probabilities of grid cells and overlap of distributions at adjacent temperatures.
%\fnote{probably somebody else should add a bit more discussion here, and maybe revise the last few paragraphs}

\subsubsection{Large negative eigenvalues: nonconvex variational problem and importance sampling.} 

\paragraph{Warmup example.} To explain our method of handling large negative eigenvalues, it helps to start with a much easier special case of the argument. Consider
$p_{-ww^\top,0}(\sigma) \propto e^{-\langle w, \sigma \rangle^2/2}$
for $\sigma \in \{\pm 1\}^n$, i.e., a rank one Ising model with interaction matrix $-ww^\top $. We claim that we can sample from $P_{-ww^\top,0}$ using \emph{rejection sampling}: %. In other words, the following sampling algorithm is computationally efficient: 
(1) first, sample $\sigma_0 \sim \mathsf{Uniform}(\{\pm 1\}^n)$, and then (2) with probability $e^{-\langle w, \sigma_0 \rangle^2/2}$ output $\sigma = \sigma_0$, and otherwise restart with step (1). From the definition, it's clear that this process draws a sample from $P$; the only concern is how long it takes. The runtime is a geometric random variable with parameter $p = \E_{\sigma_0} e^{-\langle w, \sigma_0 \rangle^2/2}$ and using Jensen's inequality we have
$p \ge e^{-\E_{\sigma_0} \langle w, \sigma_0 \rangle^2/2} = e^{-\|w\|^2/2}$.
%where we used that the covariance matrix of $Uni \{\pm 1\}^n$ is identity in the last equality;
Hence, the expected runtime is $1/p = \exp(\|w\|^2/2)$ (constant time provided $\|w\| = O(1)$).

This is an artificially simple example because: (1) the Ising model we considered had no positive eigenvalues, and (2) there was no external field. In all of the cases of serious interest, %if we try to sample by running 
rejection sampling from the uniform distribution 
%then the runtime will be extremely bad (exponential in the dimension $n$). 
has extremely bad runtime (exponential in dimension $n$).
However, generalizing this example leads us naturally to a %(seemingly more sophisticated)
more sophisticated algorithm which works more generally.

\paragraph{The general importance sampling argument and fixed point equation.} The actual problem we need to solve is this: sample from an Ising model with external field $h$ and interaction matrix $J$ with the following structure: $J = J_+ - J_-$ with $0 \preceq J_+ \preceq 1 - c$ and $0 \preceq J_-$ with small trace.
(We use the previous annealing argument to eliminate any larger positive eigenvalues.) We will let
$ Q(\sigma) \propto e^{\rc 2\langle \sigma, J \sigma \rangle + \langle h, \sigma \rangle} $
denote the Ising model we ultimately want to sample from. 

To have any hope of succeeding with the rejection sampling approach, we need a smart proposal distribution. Since we have a sampler for the Ising model $P_{J_+,h}(\sigma) \propto e^{\rc 2\langle \sigma, J_+ \sigma \rangle + \langle h, \sigma \rangle}$, this would be an obvious choice of proposal distribution. However, this is a bad idea: the distribution $p_{J_+,h}$ and the target distribution many be concentrated around different regions\footnote{For a concrete example, suppose $J_+ = 0$, $J_- = \vec{1} \vec{1}^\top /n$ and $h = \vec{1}$. Then by explicit calculation, it can be shown that mean without the $J_-$ term is much further from zero than with the $J_-$ term included.}, in which case rejection sampling will perform poorly. A smarter choice is to consider a tilted proposal distribution with additional external field $\lambda \in \mathbb{R}^n$, i.e., an Ising model of the form 
$P_{J_+, h + \lambda}(\sigma) \propto e^{\rc 2\langle \sigma, J_+ \sigma \rangle + \langle h + \lambda, \sigma \rangle}$.
Then the relative density satisfies
$\frac{dQ}{dP_{J_+, h + \lambda}}(\sigma) \propto e^{-\rc 2\langle \sigma, J_- \sigma \rangle - \langle \lambda, \sigma \rangle}$
and if we specifically consider tilts of the form $\lambda = -J_- \mu$, we can complete the square to write
\[ \frac{dQ}{dP_{J_+, h - J_-\mu}}(\sigma) = \frac{1}{Z(\mu)} e^{-\rc 2\langle \sigma - \mu, J_- (\sigma - \mu) \rangle} \]
where $Z(\mu) := \mathbb{E}_{P_{J_+, h - J_-\mu}}[e^{-\rc 2\langle \sigma - \mu, J_- (\sigma - \mu) \rangle}]$ is the normalizing constant. Note that $Z(\mu) \le 1$ since $J_-$ is positive semidefinite. To lower bound $Z(\mu)$, analogous to the ``warmup example,'' we can apply Jensen's inequality, which gives
\begin{equation}\label{eqn:jensen-application}
\log Z(\mu) \ge -\mathbb{E}_{P_{J_+, h - J_-\mu}}[\langle \sigma - \mu, J_- (\sigma - \mu) \rangle/2] = -\langle J_-, \mathbb{E}_{P_{J_+, h - J_-\mu}}[(\sigma - \mu)(\sigma - \mu)^\top ] \rangle. 
\end{equation}
For arbitrary $\mu$, the right hand side of this inequality does not seem particularly tractable. \emph{However}, if were fortunate enough to choose $\mu$ which is a solution of the fixed point equation
\begin{equation}\label{eqn:fixedpoint}
\mu = \mathbb{E}_{P_{J_+, h - J_-\mu}}[\sigma]    
\end{equation}
then on the right hand side of \eqref{eqn:jensen-application}, the term $\mathbb{E}_{P_{J_+, h - J_-\mu}}[(\sigma - \mu)(\sigma - \mu)^\top ]$ \eqref{eqn:jensen-application} is simply a covariance matrix. Because $P_{J_+, h - J_-\mu}$ is an Ising model with all eigenvalues lying in an interval of length $1 - c$, its covariance matrix is bounded in operator norm by $1/c$ \citep{eldan2020spectral}. Hence by the matrix H\"older inequality, we have
$\log Z(\mu) \ge -\langle J_-, \mathbb{E}_{P_{J_+, h - J_-\mu}}[(\sigma - \mu)(\sigma - \mu)^\top ] \rangle \ge -\frac{1}{c}\Tr(J_-)$.
Provided such a $\mu$ exists, this lets us perform importance sampling with expected running time $e^{(1/c) \Tr(J_-)}$, by using $P_{J_+, h - J_-\mu}$ as the proposal distribution, which we can sample from using Glauber dynamics by Theorem~\ref{thm:glauber}.

\paragraph{Solving the fixed point equation: variational argument and nonconvex SGD.} There is only one problem remaining: how do we find a solution of the fixed point equation \eqref{eqn:fixedpoint}, or even know that one exists? To show existence, we use what is known as a \emph{variational argument}: we construct a functional $G(\mu)$ and prove that (1) any critical point of $G$ solves our desired equation \eqref{eqn:fixedpoint}, and (2) $G$ has at least one global minima, hence at least one critical point. This strategy is quite familiar in the context of variational inference (e.g., constructing BP fixed points \citep{mezard2009information}),  as well as in other fields in mathematics like classical mechanics and PDEs \citep{evans2010partial}.

In our case, we can first assume $J_-$ is strictly positive definite without loss of generality (by adding a small copy of the identity to $J_-$, which preserves the distribution and only slightly increases the trace). Then we consider the functional
\begin{equation}
    G(\mu) := \log \mathbb{E}_{P_{J_+,h}}[e^{\langle \mu, -J_- \sigma \rangle}] +\rc2 \langle \mu, J_- \mu \rangle.
\end{equation}
Differentiating, we obtain
\begin{equation}\label{eqn:gradient-intro}
\nabla G(\mu) = - J_-\mathbb{E}_{P_{J_+,h - J_- \mu}}[\sigma] + J_-\mu
\end{equation}
and because $J_-$ is invertible, this means that $\nabla G(\mu) = 0$ iff $\mu$ solves the fixed point equation \eqref{eqn:fixedpoint}.  

To show there exists a global minimizer of $G(\mu)$, we observe that $G(0) = 0$ and by H\"older's inequality that
$G(\mu) \ge - \|J_- \mu\|_1 + \langle \mu, J_- \mu \rangle/2 \ge -\sqrt{n} \|J_-\|_{\op} \|\mu\| + \langle \mu, J_- \mu \rangle/2$.
The first negative term grows at most linearly in $\|\mu\|$, whereas the second positive term grows quadratically in $\|\mu\|$ because $J_-$ is positive definite. Thus, for all $\mu$ with $\|\mu\|$ sufficiently large, we must have that $G(\mu) > 0$. Hence the infimum of $G$ must be achieved within a compact ball around $0$, and so $G$ has at least one global minima and at least one critical point.

Now that we have shown that a fixed point exists, there is a clear way to make this argument constructive: run stochastic gradient descent to try to minimize $G(\mu)$, starting from zero. Based on \eqref{eqn:gradient-intro}, we can indeed compute a stochastic gradient of $G$ provided we can sample from $P_{J_+,h - J_- \mu}$, which we do via Glauber dynamics (Theorem~\ref{thm:glauber}). While SGD is not guaranteed to find the global minimum, we can use the result of \citet{ghadimi2013stochastic} to guarantee that SGD at least finds an approximate critical point, which is sufficient.% and completes the argument. 

\paragraph{The general case: Positive and negative eigenvalues.} We now describe how to combine the techniques to deal with general case when $J = J^+-J^-$ can have both positive and negative eigenvalues. In the PSD case, we computed the partition function for~\eqref{e:psimuprl} over a grid of $\mu^\prl$'s. We cannot include the negative definite part in $J^\perp$, but we know from our variational argument that we can approximate $p_{J^\perp - J_-, h+X^\top \mu^\prl}$ with $p_{J^\perp , h+X^\top \mu^\prl + f(\mu^\prl)}$ for some $f(\mu^\prl)$ we can compute; hence we run the annealing and tempering argument on these distributions instead, with a final step of importance/rejection sampling to bring us back to $p_{J^\perp - J_-, h+X^\top \mu^\prl}$.

%% file: ack.tex
\section*{Acknowledgements}

This work was done in part while the authors were visiting the Simons Institute for the Theory of Computing.

%% file: notations.tex
We collect here some notation used in the paper for easy reference.
\paragraph{Probability distributions and partition functions.}
\begin{align*}
    p_{J,h} (\si) &= \rc{Z_{J,h}} \ising{J}{h}\\
    Z_{J,h} &= \sum_{\si\in \BC} \ising{J}{h}\\
    \pproj^{\si,\mu^\prl} &\propto p_{J^\perp, h+X^\top \mu^\prl}(\si) \exp\pa{-\fc n2 \ve{\mu^\prl}^2}\\
    &= \exp \pa{\isingexp{J^\perp}{h+X^\top \mu^\prl} - \fc n2 \ve{\mu^\prl}^2}\\
    \pproj^{\si,y}(\si,y) &= \pproj^{\si,\mu^\prl}(\si, Qy)\\
    \Zproj(\mu^\prl) &= Z_{J^\perp, h+X^\top \mu^\prl} \exp\pa{-\fc n2 \ve{\mu^\prl}^2}\\
    \Zproj &= \int_{V^\prl} \Zproj(\mu^\prl)\,d\mu^\prl 
\end{align*}
\paragraph{Decomposing $J$.}
\begin{align*}
    J &= J_+-J_-\\
    J_+ &= \rc n XX^\top\\
    J^\prl &= \rc n X^\top P^\prl X\\
    J^\perp &= \rc n X^\top P^\perp X = J_+-J^\prl\\
    \Jpall &= J^\perp - J_-\\
    V &= \text{ subspace of $\R^n$ spanned by eigenvectors of $J_+$ with eigenvalues $>1-\rc c$}\\
    Q &= \text{ $n\times d$ matrix whose columns are an orthogonal basis for $V$}
\end{align*}
\paragraph{Probability distributions, partition functions, and partition function estimates from annealing/tempering.}
\begin{align*}
\Grid &= \bc{-L + \rc 2 \eta, -L + \fc 32 \eta, \ldots, L-\rc 2\eta}^d\\
\mu(y^*) &= \text{ approximate critical point of }G(u)=\log  \E_{\si\sim P_{J^\perp, X^\top  Qy^*}}[e^{-\an{u, J_-\si}}] + \rc 2 \an{u, J_-u}\\
h(y^*) &= \mu(y^*) + X^\top Qy^* + h\\
B(y^*)& =  \text{ hypercube with sides parallel to the standard axes, centered at $y^*$ with side length $\eta$}\\
    p_{\ell,y^*} &= p_{\be_\ell J^\perp, h(y^*)} \text{ where }\be_\ell = \fc{\ell-1}n\\
    p_{M+1}(\si,y^*) &=\fc{\int_{B(y^*)} \exp\pa{\rc 2 \an{\si, \Jpall \si} + \an{X^\top Qy + h, \si} - \fc n 2\ve{y}^2}\,dy}{\int_{[-L,L]^d}\sum_{\si\in \{\pm 1\}^d} \exp\pa{\rc 2 \an{\si, \Jpall \si} + \an{X^\top Qy + h, \si} - \fc n2\ve{y}^2}\,dy} \\
    g_\ell(\si) &= \exp\pa{\rc2\pa{\be_{\ell+1}-\be_\ell}\an{\si,J^\perp \si}} = \exp\pa{\rc{2n}\an{\si, J^\perp \si}}, \quad 1\le \ell \le M-1\\
    g_M(\si) = g_{M,y^*}(\si) &=\fc{\exp\pa{-\rc 2 \an{\si, J_-\si}} 
        }{\exp\pa{\an{\mu(y^*),\si}}}\int_{B(y^*)} \exp\pa{\an{X^\top Q(y-y^*), \si} 
        -\fc n2 \ve{y}^2} \,dy\\
    \wh Z_\ell(y^*) &= \text{ estimate for }Z_\ell(y^*)\\
    \wh Z(y^*) &= \wh Z_{M+1}(y^*)\\
    Z_\ell(y^*) &= Z_{\be_\ell J^\perp, h(y^*)}\\
    R_\ell(y^*) &= \fc{Z_\ell(y^*)}{\wh Z_\ell(y^*)}\\
    q_{\ell,y^*} &= \fc{Z_\ell(y^*)}{\wh Z_\ell (y^*)} p_{\ell,y^*}\\
    \pst_\ell(\si,y) &= \pa{\wh Z_\ell(y^*) \sum_{y\in \Grid} R_\ell(y)}^{-1} \ising{J^\perp}{h(y^*)}.
\end{align*}

%% file: negdef.tex
\section{Sampling with negative definite spikes using a variational argument}
\label{sec:negdef}
The proof of the following result gives a generic algorithm which, given sampling access to a distribution $P$ and its tilts, samples from any distribution $Q$ which is reweighted by a negative definite quadratic form with small trace. As stated, the result applies to any distribution supported on a $\sqrt{n}$-radius sphere, not just discrete distributions on the hypercube.
In fact, when $-J$ is strictly negative definite, the exact same argument applies not just to distributions on the sphere, but supported on any compact set. 
%The reduction is based on a variational argument, algorithmically implemented using stochastic gradient ascent on a nonconvex function, which seems to be new.

\begin{thm}\label{thm:negdef}
Suppose we are given a sampling oracle for a distribution $P$ 
supported on the sphere $\{ x : \|x\| = \sqrt{n} \}$
%on the hypercube $\{\pm 1\}^n$ 
and all of its tilts 
\[ \frac{dP_{\lambda}}{dP}(x) \propto e^{\langle \lambda, x \rangle}. \]
Also, suppose that for any $\lambda$ the covariance matrix of $P_{\lambda}$ is upper bounded in spectral norm by $M$.
Then for any $J \succeq 0$ and $\ep > 0$, if we define the reweighted measure
\[ \frac{dQ}{dP}(x) \propto e^{-\langle x, J x \rangle/2}, \]
then there exists an algorithm which with probability at least $1 - \delta$, outputs $\lambda \in \mathbb{R}^n$ such that
\[ \log \frac{dQ}{dP_{\lambda}}(x) \le M\Tr(J) + \ep \]
with runtime and oracle complexity polynomial in $n$, $1/\ep$, $M$, $\log(1/\delta)$, and $\|J\|_{\op}$. 
\iffalse a sample from 
\[ Q(x) \propto e^{-\langle x, J x \rangle/2} P(x) \]
with runtime and oracle complexity polynomial in $e^{M (1 + \Tr J)}, n$.
\hlnote{In particular, this holds for Ising models with PSD interaction matrix $\ve{J_+}_\op\le 1-\rc c$...}
\fi
\end{thm}
Specializing this result to the case of Ising models gives the following algorithmic result.
\begin{cor}\label{thm:negdef-ising}
Suppose that $J$ is an arbitrary symmetric matrix and decompose $J = J_+ - J_-$ where both $J_+,J_-$ are positive semidefinite and suppose that $\ve{J_+}_\op \le 1 - \rc c$ for $c > 0$. Let $h \in \mathbb{R}^n$ be arbitrary, and define $Q(\sigma) \propto \exp(\rc2 \langle \sigma, J \sigma \rangle + \langle h, \sigma \rangle)$ and $P_{\lambda}(\sigma) \propto \exp(\rc2\langle x, J_+ x \rangle + \langle h + \lambda, x \rangle)$. 
There exists an algorithm which with probability at least $1 - \delta$, outputs $\lambda \in \mathbb{R}^n$ such that
\[ \log \frac{dQ}{dP_{\lambda}}(\sigma) \le c\Tr(J_-) + \ep \]
with runtime and oracle complexity polynomial in $n$, $1/\ep$, $M$, $\log(1/\delta)$, and $\|J_-\|_{\op}$. 
%\hlnote{Is this supposed to be $\Tr(J_-)$?} % those are polynomially equivalent
\end{cor}
\begin{proof}
This follows by applying Theorem~\ref{thm:negdef} with $\delta' = \delta/2$.
First, we recall from \cite{eldan2020spectral} (as a consequence of the Poincar\'e inequality) that we can take 
$M = \frac{1}{1 - \|J_+\|_{\op}} \le c$ where $M$ is the upper bound on the spectral norm of the covariance matrix of $P_{\mu}$ as defined in Theorem~\ref{thm:negdef}. If we supposed we had access to an exact sampler from each of the distributions $P_{\mu}$, this would imply the result. Since we instead will implement each sampling call with a Markov chain (the Glauber dynamics) which can draw samples extremely close to the distribution $P_{\mu}$, the actual result follows by coupling these outputs to a hypothetical process which has exact samples.

More precisely, from Theorem~\ref{thm:glauber} we can draw a sample from any of the distributions $P_{\lambda}$ in polynomial time in the sense that for any $\ep > 0$, with $\poly(n,\log(1/\ep))$ time we can generate a sample with total variation distance at most $\ep$. If $q$ is the maximum number of queries made by the algorithm from Theorem~\ref{thm:glauber}, then by taking $\ep = \delta/2q$ and using the union bound, we can with probability at least $1 - \delta/2$ couple all of the outputs of the Markov chains invoked at every oracle call with samples from the true distribution $P_{\lambda}$. Therefore, with total probability at least $1 - \delta$, the algorithm which uses Markov chain samplers will output $\lambda$ satisfying the guarantee of Theorem~\ref{thm:negdef}. This proves the result. 
\end{proof}

%\fnote{swap statements}
We now proceed to the proof of Theorem~\ref{thm:negdef}. 
In the algorithm and analysis, we will use the fact that stochastic gradient descent with an appropriate step size schedule is able to find approximate critical points of smooth functions (a stronger and more explicit result is given in the original statement in \cite{ghadimi2013stochastic}, see also \cite{allen2018make}).
\begin{thm}[Corollary 2.5 of \cite{ghadimi2013stochastic}]\label{thm:sgd}
Suppose that $f$ is a differentiable function which is $L$-smooth with respect to the Euclidean norm $\|\cdot\|$ in the sense that for all $x,y$
\[ \|\nabla f(x) - \nabla f(y)\| \le L\|x - y\|. \]
Let $f^* := \inf_x f(x)$ and define 
\[ D_f := \sqrt{\frac{2(f(x_1) - f^*)}{L}}. \]
Then there exists a polynomial time algorithm (2-RSG, the two-phase randomized stochastic gradient algorithm) which given oracle access to (identical, independent copies of) a stochastic gradient oracle $g$ such that $\mathbb{E}[g(x_t) \mid x_t] = \nabla f$ and $ \mathbb{E}[\exp(\|g(x_t)\|^2/\sigma^2) \mid x_t] \le 1$ and $\ep > 0$, with probability at least $1 - \delta$ outputs $x$ such that $\|\nabla f\| \le \ep$ using $\poly(D_f,\log(1/\delta),\sigma,L,1/\ep)$ runtime and oracle calls.
\end{thm}
%\fnote{check what the actual run time is in $n$} \hlnote{Say that the algorithm still gives an approximate stationary point when there is $1/\poly$ bias in the gradient} \fnote{I think you can just use a coupling argument rather than having to change this to handle bias? (i.e. couple with the algorithm which only gets samples from the true distribution)}

%\fnote{change statement to output $\mu$ and state radon-nikodym derivative bound (make $1$ smaller)}
\begin{proof}[Proof of Theorem~\ref{thm:negdef}]
First, we can assume $J \succeq \ep/Mn$ without loss of generality by adding $(\ep/Mn)I$ to $J$, which does not change the measure $Q$ and increases the trace by just $\ep/M$. (This only changes the final guarantee by an additional additive $\ep$, which can be trivially corrected by dividing $\ep$ by $2$.) %\fnote{fixing the signs below, flip $J$ to $-J$}

The key idea of the proof is a variational argument. Define the functional
\[ G(\mu) := \log \mathbb{E}_{P}[e^{\langle \mu, -J X \rangle}] + \rc2\langle \mu, J \mu \rangle \]
and observe that its derivative can be expressed in terms of the tilted measure $P_{-J \mu}$:
\[ \nabla G(\mu) = -J \mathbb{E}_{P_{-J \mu}}[X] + J \mu. \]
Now observe that for any $\mu$,
\[ \frac{dQ}{dP_{-J \mu}}(x) \propto e^{-\langle x, J x \rangle/2 + \langle J \mu, x \rangle} \propto e^{-\rc2\langle x - \mu, J (x - \mu) \rangle} \]
and so
\[ \frac{dQ}{dP_{-J \mu}}(x) = \frac{1}{Z} e^{-\rc 2\langle x - \mu, J (x - \mu) \rangle} \]
where
\[ Z := \mathbb{E}_{P_{J \mu}}[e^{-\rc2\langle X - \mu, J (X - \mu) \rangle}]. \]
From the definition and the fact that $J$ is psd, we have $Z \le 1$. Also, by Jensen's inequality
\[ Z = \mathbb{E}_{P_{J \mu}}[e^{-\rc2\langle X - \mu, J (X - \mu) \rangle}] \ge \exp\left(-\mathbb{E}_{P_{J \mu}}\ba{\rc2\langle X - \mu, J (X - \mu) \rangle}\right). \]
Observe that if $\Sigma := \mathbb{E}_{P_{J \mu}}[XX^\top ] - \mathbb{E}_{P_{J \mu}}[X]\mathbb{E}_{P_{J \mu}}[X]^\top $ then
\begin{align*}
\mathbb{E}_{P_{J \mu}}[(X - \mu)(X - \mu)^\top ]
&= \mathbb{E}_{P_{J \mu}}[XX^\top ] - \mathbb{E}_{P_{J \mu}}[X]\mu^\top  - \mu\mathbb{E}_{P_{J \mu}}[X]^\top  + \mu \mu^\top  \\
&= \Sigma + \mathbb{E}_{P_{J \mu}}[X]\mathbb{E}_{P_{J \mu}}[X]^\top  - \mathbb{E}_{P_{J \mu}}[X]\mu^\top  - \mu\mathbb{E}_{P_{J \mu}}[X]^\top  + \mu \mu^\top  \\
&= \Sigma + (\mathbb{E}_{P_{J \mu}}[X] - \mu)(\mathbb{E}_{P_{J \mu}}[X] - \mu)^\top 
\end{align*}
so 
\begin{align} 
\nonumber
\log Z 
&\ge -\rc2\langle \mathbb{E}_{P_{J \mu}}[(X - \mu)(X - \mu)^\top ], J\rangle \\
\nonumber
&= -\rc2\langle \Sigma + (\mathbb{E}_{P_{J \mu}}[X] - \mu)(\mathbb{E}_{P_{J \mu}}[X] - \mu)^\top , J \rangle \\
\nonumber
&\ge -\rc2\|\Sigma\|_{\op} \Tr(J) - \rc2\|J\mathbb{E}_{P_{J \mu}}[X] - J\mu\|_2 \|\mathbb{E}_{P_{J \mu}}[X] - \mu\|_2 \\
\nonumber
&= -\rc2\|\Sigma\|_{\op} \Tr(J) - \rc2\|\nabla G(\mu)\|_2  \|\mathbb{E}_{P_{J \mu}}[X] - \mu\|_2.
\end{align}
Note that the final lower bound can be maximized if we can find a critical point of $G$. We next argue that such a critical point exists.

Note that $G(0) = 0$ by definition and because we reduced to the case $J \succeq (\ep/Mn)I$, 
\begin{align} \label{eqn:Gmu-bound}
G(\mu) 
&\ge \log \mathbb{E}_P[e^{\langle \mu, -J X \rangle}] + (\ep/2Mn)\|\mu\|_2^2 \nonumber \\
&\ge -\|J \mu\|_1 + (\ep/2Mn)\|\mu\|_2^2 \ge -\|J\|_{\op} \|\mu\|_2\sqrt{n} + (\ep/2Mn)\|\mu\|_2^2 ,
\end{align}
which is positive provided $\|\mu\|_2 > 2Mn^{3/2}\|J\|_{\op}/\ep$. Hence the global minimum of $G$ must be attained somewhere on the compact set $\mathcal{K} = \{\mu : \|\mu\|_2 \le 2Mn^{3/2}\|J\|_{\op}/\ep \}$. At this point, we have proved the existence of a critical point. We next show that one can be approximately found with stochastic gradient descent initialized at zero, by checking the assumptions of Theorem~\ref{thm:sgd}. 
%(Note that we will apply the theorem to the negation of $G$, because we will show the gap in objective value from initialization to the global maximum is bounded, therefore the algorithm is performing stochastic gradient ascent on $G$ and not descent.)

By the invertibility of $J$, any solution of the equation $0 = \nabla G(\mu) = -J \mathbb{E}_{P_{J \mu}}[X] + J \mu$ satisfies $\mu = \mathbb{E}_{P_{J \mu}}[X]$ and hence $\mu \in [-1,1]^n$ and $\|\mu\|_2 \le \sqrt{n}$. In particular the global minimum satisfies this, so combined with \eqref{eqn:Gmu-bound} we have
\[ \inf_{\mu} G(\mu) \ge \inf_{r \le \sqrt{n}}[-r\sqrt{n} \|J\|_{\op} + (\ep/2Mn)r^2] > -% n \|J\|_{\op}. 
\fc{Mn^2}{2\ep}\ve{J}_{\op}^2 %HL fixed
\]
Since
\[ \nabla^2 G(\mu) = -J \mathbb{E}_{P_{J \mu}}[XX^\top ] J + J \]
we have that $\|\nabla^2 G(\mu)\|_{\op} \le M\|J\|_{\op}^2 + \|J\|_{\op} =: L$ %/2
which means that $G(\mu)$ is $L$-smooth with respect to the Euclidean norm. 
Recalling that $\nabla G(\mu) = -J \mathbb{E}_{P_{J \mu}}[X] + J \mu$, we see that if $x \sim P_{J \mu}$, which we have a sampling oracle for by assumption, then $g(\mu) := -J (x - \mu)$ is a stochastic gradient oracle for $G(\mu)$ satisfying $\|g(\mu)\| \le \|J\|_{\op} \|x - \mu\| \le 2\|J\|_{\op} \sqrt{n}$. This means that all of the assumptions of Theorem~\ref{thm:sgd} are satisfied and we can find an $\ep$-approximate critical point of $G$ using $\poly(n,\|J\|_{\op},M,1/\ep)$ runtime and calls to the sampling oracle. Outputting $\lambda := -J \mu$ gives the result. %$\poly(e^{M(1 + \Tr J)}, n)$ runtime and calls to the sampling oracle. 
\iffalse Also, by optimizing the right hand side of \eqref{eqn:Gmu-bound} we see that
\[ G(\mu) \le \sup_{r \ge 0} [r\sqrt{n} \|J\|_{\op}  - r^2/2n] = n^2\|J\|_{\op}^2/2. \]
Recalling that $\nabla G(\mu) = J \mathbb{E}_{P_{J \mu}}[X] - J \mu$, we see that if $x \sim P_{J \mu}$ then $g(\mu) := J (x - \mu)$ is a stochastic gradient oracle for $G(\mu)$ satisfying $\|g(\mu)\| \le \|J\|_{\op} \|x - \mu\|$
\fi
%and so
%\[ \sup_{\mu} G(\mu) = \sup_{\mu \in \mathcal{K}} G(\mu) \le 2\|J\|_{\op} n^2\]. 
%Also, this implies gradient ascent initialized at 0 with sufficiently small step size (TODO, use smoothness) will converge to an approximate critical point within this set in a polynomial number of steps, and a sufficiently accurate approximation of the gradient can be computed at each step with a polynomial number of queries to the oracle.
%Since the global maximum has to be a critical point, it must satisfy $\nabla G(\mu) = 0$ which means $\mu = \mathbb{E}_{P_{J_{\mu}}}[X] \in [-1,1]^n$, i.e. the global maximum is attained at some $\mu$ in the closed hypercube $[-1,1]^n$.
%= \frac{e^{-\langle x, J x \rangle/2 - \langle J \mu, x \rangle}}{\mathbb{E}_{P_{J \mu}}[e^{-\langle x, J x \rangle/2 - \langle J \mu, x \rangle}]} \]

%Finally, after the approximate critical point $\mu^*$ is found, it follows from the above that we can draw samples from $Q$ using importance sampling with sample complexity polynomial in $e^{\|\Sigma\|_{\op} \Tr(J)/2}$, since this upper bounds $1/Z$ and hence the relative density to $P_{J \mu^*}$.
\end{proof}

%\hlnote{Show using Poincar\'e inequality that $\ve{\Si}_\op\le c$ when eigenvalues are $\le 1-\rc c$.}

\begin{rem}\label{rem:variational-comparison}
The variational argument in the proof is partially inspired by, thought different from, some previous arguments in the variational methods literature; for example, the construction of Belief Propagation fixed points using the Bethe free energy, and variants of this argument which arise from the Thouless-Anderson-Palmer and naive mean-field free energy (see e.g., \cite{mezard2009information,wainwright2008graphical}). %\hlnote{this sentence is hard to parse.}
As with all such variational arguments, the key idea is to construct a solution to a fixed point equation by writing it as the gradient of a well-behaved functional.
To make a more explicit connection with that literature, consider the special case where $P$ is a product measure on the hypercube $\{ \pm 1\}^n$, so $P(\sigma) \propto e^{\langle h_0, \sigma \rangle}$ for some $h_0 \in \mathbb{R}^n$ encoding the bias of each coordinate. Then the equation $\nabla G(\mu) = 0$ is equivalent to $-J \mu = -J \tanh(h_0 - J \mu)$
and because $J$ is invertible, it simplifies to the fixed-point equation
\[ \mu = \tanh(h_0 - J \mu). \]
This is almost the same as the naive mean-field fixed point equation, except that in that case, the diagonal of $J$ must be zeroed out whereas in our case they are not. Relatedly, $G(\mu)$ is not the same as the naive mean-field free energy corresponding to $Q$, and the positive definiteness of $J$ is not needed to solve the naive mean-field equations but plays a key role in our variational argument.
%, which would be $F(\mu) = -\langle \mu, [J - diag(J)] \mu \rangle/2 - \sum_{i = 1}^n KL(Ber_{\pm}(\mu), Ber_{\pm}(\tanh h_0))$.
%\fnote{add a bunch of references}
%Still, the general case of the argument can be thought of by analogy as a way to solve a more complex fixed point equation in $\mu$.
%% Looks weird, think if this is meaningful
%For another connection, by the Gibbs variational principle we can write
%\[ G(\mu) = \sup_{R} \left[\langle J \mu, \E_R[X] \rangle - KL(R || P) \right] - \langle \mu, J \mu \rangle/2. \]
%\fnote{compare to variational inference: Bethe approximation and mean field}
\end{rem}
%\fixme{Integrate this into estimation and sampling sections}
%% \subsection{Sampling with positive and negative spikes}

% %We first obtain 2-approximations to $Z_{J_+,h(y^*)+\mu(y^*)}$ and 
% We temper using the distributions $p_{J_+,h(y^*)+\mu(y^*)}$ and include an extra $\exp\pa{-\fc{c\Tr(J_-)}2}$ factor in the acceptance ratio. 
% \hlnote{Add details.}

%% file: Z.tex
\section{Estimating the partition function}\label{sec:Z}
In this section, we develop and analyze an algorithm for computing the partition function $Z$. 

\paragraph{Application of the Hubbard-Stratonovich transform.}
Based on the Hubbard-Stratonovich transform, we can easily prove the following Theorem. (We warn the reader that the notation has a couple minor cosmetic differences from the Technical Overview, with the goal of minimizing ambiguity.)
%Its purpose is explained in the Technical Overview (which the reader may want to read first).

\begin{thm}\label{t:hs}
Let $J\in \R^{n\times n}$ be a symmetric matrix, and write $J = \rc n X^\top X - J_-$ for $X\in \R^{m\times n}$ and $J_-$ negative semi-definite. 

Let $V\subeq \R^m$ be a subspace. Let $P^\prl$ and $P^\perp$ be the projections onto $V$ and $V^\perp$. Let $J^\prl = \rc n X^\top P^\prl X$ and $J^\perp = J - J^\prl$. %\rc n X^\top P^\perp X$. 
Then
\begin{align*}
Z_{J,h} &= \pf{n}{2\pi}^{d/2} \Zproj &
\text{where}\quad 
\Zproj &= \int_{V^\prl} Z_{J^\perp, h + X^\top  \mu^\prl} \exp\pa{-\fc n2 \ve{\mu^\prl}^2}\,d\mu^\prl.
\end{align*}
\end{thm}
Note that in the special case that $V=\R^m$ and $J_-=O$, this gives a decomposition of the probability measure in terms of product distributions in a similar manner to \citep{bovier2012mathematical,bauerschmidt2019very}. 

\begin{proof}[Proof of Theorem~\ref{t:hs}]
We decompose $J=J^\perp + J^\prl$ and apply Lemma~\ref{l:hs} to $X\leftarrow P^\prl X$ with $\gamma^2 = 1/n$:
\begin{align*}
    \ZJh &= \sum_{\si\in \{\pm 1\}^n} \eJhs \\
    &= \sum_{\si\in \BC} \exp\pa{\isingexp{J^\perp}h} \exp\pa{\rc{2n} \ve{P^\prl X \si}^2}\\
    &= \pf{n}{2\pi}^{d/2} \sBC \ising{J^\perp}h \int_{V^\prl} \exp\pa{\an{X^\top P^\prl \mu^\prl , \si} - \fc n2 \ve{\mu^\prl}^2}\,d\mu^\prl \\
    &= \pf{n}{2\pi}^{d/2} \int_{V^\prl} Z_{J^\perp, h+X^\top \mu^\prl}  \exp\pa{-\fc n2 \ve{\mu^\prl}^2} \,d\mu^\prl,
\end{align*}
as desired.
\end{proof}
We can define an associated probability distribution on $\BC\times V$ with $\Zproj$ as its partition function:
\[
\pproj ^{\si,\mu^\prl}(\si,\mu^\prl) \propto 
\exp\pa{\isingexp{J^\perp}{h+X^\top \mu^\prl} - \fc n2 \ve{\mu^\prl}^2}.
\]
Choosing an orthogonal linear transformation $Q:\R^d \to V$, we will also define the distribution $\pproj ^{\si,y}(\si,y) = \pproj ^{\si,\mu^\prl}(\si,Qy)$. 
In Appendix~\ref{app:equiv}, we will interpret %the associated probability distribution of $(\si,\mu)$ on the extended state space $(\si, \mu)\in \BC \times \R^m$ 
$\pproj ^{\si,y}$ as the posterior of a Gaussian mixture model after seeing samples given by the columns of $X$. 
%$\mu$ in the Hubbard-Stratonovich transform as the mean of

\paragraph{Estimating the partition function.}
For a PSD matrix $A$, let $\rank_\tau(A)$ denote the number of eigenvalues of $A$ that are $\ge \tau$. Note that $\rank_1(A)\le \ve{A}_F^2$. For ease of exposition, we first prove the theorem when in the case where $J$ has no negative eigenvalues.
\begin{thm}\label{t:Z}
Let $\ep,\de\in (0,1)$. 
Suppose $J$ is PSD. 
With probability $\ge 1-\de$, Algorithm~\ref{a:Z-ising} outputs an $e^\ep$-multiplicative approximation to $Z_{J,h}$, 
\begin{align*}
e^{-\ep} Z_{J,h} \le \wh Z_{J,h} &\le e^\ep  Z_{J,h},
\end{align*}
in time $(\ve{J}_\op n)^{O(\rank_1(J)+1)} O\pa{\log \prc \de /\ep^2}$.
%$O\pa{\sqrt n \log \pf n\ep}^{\rank_1(J)} \cdot O\pf{\poly(n) \log \prc{\ep\de}}{\ep^2}$.
%$\fc{(\ve{J}_2n)^{O(1+\rank_1(J))}\log \prc \de}{\ep^2}$.
\end{thm}

Given a probability distribution $p$ on $\{\pm 1\}^n$, we can define the Markov chain in Algorithm~\ref{a:mc}. For $\si\in \{\pm 1\}^n$, we let $\si^{(i)}=(\si_1,\ldots, -\si_i, \ldots, \si_n)$ denote $\si$ but with the $i$th coordinate flipped.

 \begin{algorithm2e}[h!]
 \DontPrintSemicolon
 \caption{Glauber dynamics on $\{\pm 1\}^n$} 
 \KwIn{Query access to probability distribution $p(\si)$ on $\{\pm 1\}^n$, up to constant of proportionality; number of steps $T$.}
 %\Ensure Approximate sample from $p$.
 %Approximation of partition function $Z_{J,h}$.
 \For{$1\le t\le T$} {
     Choose a random coordinate $i$, and set $\si \leftarrow \si^{(i)}$ with probability $\fc{p(\si^{(i)})}{p(\si^{(i)})+p(\si)}$.
     }

%  \begin{algorithmic}[1]
%  \Require Query access to probability distribution $p(\si)$ on $\{\pm 1\}^n$, up to constant of proportionality; number of steps $T$.
%  %\Ensure Approximate sample from $p$.
%  %Approximation of partition function $Z_{J,h}$.
%  \For{$1\le t\le T$} 
%     \State Choose a random coordinate $i$, and set $\si \leftarrow \si^{(i)}$ with probability $\fc{p(\si^{(i)})}{p(\si^{(i)})+p(\si)}$.
%  \EndFor
%  \end{algorithmic}
  \label{a:mc}
 \end{algorithm2e}

The following lemma gives %the mixing time of
fast mixing of Glauber dynamics for the Ising model, when the spectral norm of the interaction matrix is at most 1.

\begin{lem}\label{l:ls-ising}
Suppose $J\in \R^{n\times n}$ is symmetric and PSD with $\ve{J}_\op\le 1$. Then the modified log-Sobolev constant $\CMLS$ for $P_{J,h}$ is at most $e^{1/2}n^2$, and the mixing time is bounded by $O(n^2 \log  n)$.
\end{lem}

\begin{proof}
For a symmetric matrix with diagonalization $A=UDU^\top$, let $D_{\le \tau}$ denote $D$ with the entries $\ge \tau$ replaced by $\tau$, and $A_{\le \tau} := UD_{\le \tau} U^\top$.
By Theorem~\ref{thm:glauber}, the modified log-Sobolev constant for 
\[
p_{J_{\le 1-\rc n}, h}(\si)\propto 
\exp\pa{\rc 2\an{ \si, J_{\le 1-\rc n} \si} +\an{h,\si}}
\]
is bounded by $n\pa{1-\ve{J_{\le 1-\rc n}}_\op}^{-1} = n^2$. Since 
\begin{align}\nonumber
\log  \pf {p_{J,h}(\si)}{p_{J_{\le 1-\rc n}, h}(\si)}
- \log  \pf{Z_{J_{\le 1-\rc n},h}}{Z_{J,h}} 
&= \rc 2\an{\si, (J-J_{\le 1-\rc n})\si} \\
&\in 
\rc 2 \ve{J - J_{\le 1-\rc n}}_2 n \cdot [0,1] \subeq  \ba{0, \rc 2},\label{e:gibbs-r1}
\end{align}
by the Holley-Stroock perturbation lemma, the modified log-Sobolev constant for $p_{J,h}$ is bounded by $e^{1/2}n^2$. 

Finally, the exchange property holds for $p_{J,h}$ by \cite[Lemma 37]{anari2021aentropic}, so by \cite[Lemma 36]{anari2021aentropic}, the mixing time is bounded by $O((n+\CMLS)\log  n) = O(n^2\log  n)$.
\end{proof}
 
Lemma~\ref{l:ls-ising} implies that Glauber dynamics gives an efficient algorithm for sampling in our setting. To obtain an algorithm for partition function estimation, we use simulated annealing. 
Simulated annealing is a generic method to obtain an algorithm for estimating a partition function $\int_{\Om} q\,d\om$, given access to sampling oracles for a sequence of distributions $p_\ell\propto q_\ell$ such that (a) $q_1$ is known, (b) for each $\ell$, $p_\ell$ and $p_{\ell+1}$ are ``close," and (c) $p_{M+1}\propto q$. 

 \begin{algorithm2e}[h!]
 \DontPrintSemicolon
 \caption{Simulated annealing for partition function estimation} 
 \KwIn{Sampling oracles for $\wt p_\ell$ (approximations to  $p_\ell\propto q_\ell$) for $1\le \ell\le M$ (distributions on $\Om$), for example, Glauber dynamics (Algorithm~\ref{a:mc}); $Z_1 = \int_{\Om} q_1\,d\om$; number of samples $N$; number of trials $R$.}
 \KwOut{Estimate of $\int_{\Om} q_{\ell}\,d\om$ for each $1\le \ell\le M+1$.}
 %Estimate of $\int_{\Om} q_{M+1}\,d\om$.
 Let $g_\ell(x) := \fc{q_{\ell+1}(x)}{q_\ell(x)}$.\;
 \For{$1\le r\le R$}{
     Let $\wh Z^r_1 = Z_1$.\;
    \For{$1\le \ell\le M$}{
         Obtain samples $x_1,\ldots, x_N\sim \wt p_\ell$.\;
         Let $\wh Y_{\ell} = \rc N \sumo kN g_\ell(x_k)$.\;
         Let $\wh Z^r_{\ell+1} = \wh Z^r_\ell \wh Y_\ell$.\;
    }
    %\State Let $\wh Z^r = Z_1\prodo \ell M \wh Y_\ell$.
 }
 \For{$2\le \ell\le M+1$}{ 
     Let $\wh Z_\ell$ be the median of $\set{\wh Z_\ell^r}{1\le r\le R}$.
 }
  \label{a:sa}
 \end{algorithm2e}

%  \begin{algorithm}[h!]
%  \caption{Simulated annealing for partition function estimation} 
%  \begin{algorithmic}[1]
%  \Require Approximate sampling oracles $\mathsf{Samp}_\ell$ for $p_\ell\propto q_\ell$ for $1\le \ell\le M$ (distributions on $\Om$), for example, Glauber dynamics (Algorithm~\ref{a:mc}); $Z_1 = \int_{\Om} q_1\,d\om$; number of samples $N$; number of trials $R$.
%  \Ensure Estimate of $\int_{\Om} q_{\ell}\,d\om$ for each $1\le \ell\le M+1$.
%  %Estimate of $\int_{\Om} q_{M+1}\,d\om$.
%  \State Let $g_\ell(x) := \fc{q_{\ell+1}(x)}{q_\ell(x)}$.
%  \For{$1\le r\le R$}
%     \State Let $\wh Z^r_1 = Z_1$.
%     \For{$1\le \ell\le M$} 
%         \State Obtain samples $x_1,\ldots, x_N\sim \mathsf{Samp}_\ell$.
%         \State Let $\wh Y_{\ell} = \rc N \sumo kN g_\ell(x_k)$.
%         \State Let $\wh Z^r_{\ell+1} = \wh Z^r_\ell \wh Y_\ell$.
%     \EndFor 
%     %\State Let $\wh Z^r = Z_1\prodo \ell M \wh Y_\ell$.
%  \EndFor
%  \For{$2\le \ell\le M+1$} 
%     \State Let $\wh Z_\ell$ be the median of $\set{\wh Z_\ell^r}{1\le r\le R}$.
%  \EndFor 
%  \end{algorithmic}
%   \label{a:sa}
%  \end{algorithm}

\begin{lem}\label{l:sa}
Let $0<\ep<1$. Suppose that $p_\ell, 1\le\ell\le M+1$ are distributions on $\Om$, and that in Algorithm~\ref{a:sa} we are given sampling oracles for $\wt p_\ell$, $1\le \ell\le M$ such that the following hold for each $1\le \ell \le M$.
\begin{enumerate}
    \item (Variance bound) $\fc{\Var_{P_\ell}(g_\ell(x))}{(\E_{P_\ell} g_\ell(x))^2} \le \si^2$.
    \item (Bias bound) $\ab{\E_{P_\ell}g_\ell(x) - \E_{\wt P_\ell}g_\ell(x)}\le \fc{\ep}{4M}$.
\end{enumerate}
Then taking $N\ge \fc{320\si^2 M}{\ep^2}$ and $R\ge 32 \log  \prc\de$, with probability $1-\de$, the output $\wh Z$ satisfies $\wh Z\in [e^{-\ep}, e^{\ep}] \cdot Z$.
\end{lem}
The proof is standard and given in the appendix.

We can now give the algorithm and proof of Theorem~\ref{t:Z}. We show that a non-adaptive temperature schedule of length $O(n)$ is sufficient for partition function estimation. Note that a shorter schedule of length $O(\sqrt n \log  n \log  \log  n)$ is possible, and can be found in $n\polylog(n)$ total queries to approximate sampling oracles %for $p_{\be}$
at the different temperatures~\citep{vstefankovivc2009adaptive}, but we use a non-adaptive schedule for simplicity. Coordinate-wise sampling is also possible, but we will need a sequence of distributions at different temperatures for our sampling algorithm.
%\hlnote{I don't know whether there is a fixed schedule of length $O(\sqrt n)$... would need to show subexponential concentration of $\si^\top J\si$ under $p_i$.}

%https://tex.stackexchange.com/questions/204649/algorithm-keeps-appearing-on-the-last-page-of-my-latex-document
\begin{algorithm2e}[!htbp]
\DontPrintSemicolon
 \caption{Approximating partition function of Ising model. (Steps in italics are only needed in presence of a negative definite spike.)} 
 \KwIn{Ising model $(J,h)$, cutoff $L$, discretization $\eta$ dividing evenly into $L$, desired accuracy $\ep$, failure probability $\de$, number of samples $N$, number of trials $R$, steps to run Markov chains $T$, threshold $c\in (1,\iy]$.}
 \KwOut{Approximation of partition function $Z_{J,h}$.}
 Suggested parameters:  %$L=\Te\pa{\sqrt n\log  \pf n\ep}$,
 $L = \Te(\sqrt{\ve{J}_\op}+1)$, 
 $\eta\le \rc{ndL + 2n\sqrt{\ve{J}_\op d}}$, 
$N = \Te \pf{M}{\ep^2}$ where $M=n+1$, $R=\Te\pa{\log  \pf{(L/\eta)^d}{\de}}$, and $T=\Te\pa{n^2\log  \pf n\ep}$. Take $c=\iy$ if $J$ is PSD.\;
  If $\ep\le 2^{-n}$, calculate $\ZJh$ by brute force.\;
  \emph{Let $J=J_+-J_-$ where $J_+$ and $J_-$ are positive semi-definite and negative semi-definite, respectively, with column spaces intersecting only in 0.}\;
  Factor $J_+=\rc n XX^\top$ for $X\in \R^{n\times n}$.\;
  Let $V$ denote the subspace of $\R^n$ spanned by the eigenvectors of $J_+$ with eigenvalues $>1-\rc c$. Let $P^\prl$ and $P^\perp$ the projections to $V$ and $V^\perp$. Let $Q\in \R^{n\times d}$ be the matrix with columns that are an orthonormal basis for $V$.\;
  Let %$J^\perp  = \fc{X^\top P^\perp X}n$ and 
  $J^\prl  = \fc{X^\top P^\prl X}n$ and $J^\perp = \fc{X^\top P^\perp X}n$.\;
 %\medskip  
 \For{$y^*  \in \Grid:= \bc{-L + \rc 2 \eta, -L + \fc 32 \eta, \ldots, L-\rc 2\eta}^d$} {
     \emph{Let $\mu(y^*)$ be an approximate critical point of $G(u)=\log  \E_{\si \sim P_{J^\perp, X^\top  Qy^*}}[e^{-\an{u, J_-\si}}] + \rc 2 \an{u, J_-u}$, found using stochastic gradient descent (Theorem~\ref{thm:negdef-ising}/\ref{thm:sgd}) with sampling oracle given by Glauber dynamics for $P_{J^\perp, X^\top  Qy^*+h}$.} (If $J_-=O$, let $\mu(y^*)=0$.)\;
     Let $B(y^*)$ denote the hypercube with sides parallel to the standard axes, centered at $y^*$ with side length $\eta$. \;
     Apply Algorithm~\ref{a:sa} to the Ising model, with sampling algorithm given by running Glauber dynamics for $T$ steps, for the following sequence of distributions ($1\le \ell\le M= n+1$):
    \begin{align*}
        p_\ell  &= p_{\fc{\ell-1}{n} J^\perp ,h(y^*)}\\
%        \text{where }
%        \be_i &= ...\\
        g_\ell (\si) &= \exp\pa{\fc{1}{2n} \an{\si, J^\perp \si}}, \quad 1\le \ell \le n\\
        %\exp(\fc{k-1}{2n} \si^\top J \si)\\
        g_{M,y^*}(\si) &=  
        \fc{\exp\pa{-\rc 2 \an{\si, J_-\si}} 
        }{\exp\pa{\an{\mu(y^*),\si}}}\int_{B(y^*)} \exp\pa{\an{X^\top Q(y-y^*), \si} 
        -\fc n2 \ve{y}^2} \,dy
        %+\fc n2 \pa{\ve{y^*}^2-\ve{y}^2}}\,dy
        \\ 
        %\fc{\exp\pa{-\rc 2 \an{\si, J_-\si}}\int_{B(y^*)}
        %\exp\pa{\sumo in \pa{-\rc 2 \ve{\si_i P^\prl x_i - Qy^*}^2 + h_i\si_i}}\,dy,
        %}{\exp\pa{h(y^*)^\top \si}}\\
        %%{\sumo in \pa{-\rc 2 \ve{\si_i P^\prl x_i - Qy^*}^2 + h_i\si_i}}}\\
        \text{where } h(y) &= 
        \mu(y) + X^\top Q y + h,
        %\pa{\an{P^\prl x_i, Qy^*}+h_i}_{i=1}^n.
    \end{align*}
    and initial partition function
    \begin{align*}
        Z_1 = Z_{O,h(y^*)} &= 2^n \prodo in \cosh\pa{\an{x_i, Qy^*}+h_i}
    \end{align*}
    to get estimates $\wh Z_\ell(y^*)$ for $1<\ell\le M+1$. Let $\wh Z(y^*):= \wh Z_{M+1}(y^*)$. 
 } 
 Return $\wh Z = \pf{n}{2\pi}^{\fc d2} %\exp\pa{\fc n2 \Tr(J^\prl)}
 \sum_{y^*  \in \Grid} \wh Z(y^*)$.
  \label{a:Z-ising}
 \end{algorithm2e}

\begin{proof}[Proof of Theorem~\ref{t:Z}]
%Suppose $0<\ep,\de<1$. 
We may assume $\ep\ge 2^{-n}$. 
Set the temperature schedule as 
%$\be_k = \fc{k-1}{d\ve{J^\perp }_2}$ for $1\le k < d\ve{J^\perp }_2+1$ and $\be_{\ce{d\ve{J^\perp }_2}+1}=1$. Let $M=\ce{d\ve{J^\perp }_2}+1$ be the length of the temperature schedule. 
$\be_\ell = \fc{\ell-1}{n}$ for $1\le \ell \le  n+1$. %and $\be_{n+1}=1$. 
Let $M=n+1$ be the length of the temperature schedule. 
We set parameters as suggested in Algorithm~\ref{a:Z-ising}.
%We will set $\eta\le \rc{ndL + 2n\sqrt d}$, 
%$N = \Te \pf{M}{\ep^2}$, $R=\Te\pa{\log  \pf{(L/\eta)^d}{\de}}$, $L=\Te\pa{\sqrt n\log  \pf n\ep}$, and $T=\Te{n^2\log  \pf n\ep}$ (the number of times we run each Markov chain). 
Then the total time complexity of the algorithm is $O\pa{\pf{2L}\eta^d MNRT}$ times the complexity of each Markov chain step, which gives complexity $O\pa{(\ve{J}_\op + 1) nd %\log \pf n\ep
}^d \cdot O\pf{\poly(n) \log \prc{\ep\de}}{\ep^2} = (\ve{J}_\op n)^{O(\rank_1(J)+1)} O\pa{\log \prc \de /\ep^2}$.

Recall that we define the distribution $\pproj^{\si,y}(\si,y) = \pproj^{\si,\mu^\prl}(\si, Qy)$.
We now fix a particular $y^*$, and write for short $g_{M}=g_{M,y^*}$.

\paragraph{Choice of ratios $g_\ell$.} 
Define $\Zproj(\mu^\prl):=Z_{J^\perp, h+X^\top P^\prl\mu^\prl}  \exp\pa{-\fc n2 \ve{\mu^\prl}^2}$. We first compute
\begin{align*}
    \E_{p_{M}} g_{M}
    &= \rc{Z_{J^\perp, h(y^*)}} \sBC 
    \ising{J^\perp}{X^\top Q y^*+h}\\
    &\quad \cdot 
    \int_{B(y^*)} \exp\pa{\an{X^\top Q (y-y^*), \si} - \fc n2 \ve{y}^2}\\
    &= \rc{Z_{J^\perp, h(y^*)}} \sBC \int_{B(y^*)} \exp\pa{\isingexp{J^\perp}{X^\top Q y+h}-\fc n2 \ve{y}^2}\,dy \\
    &= \fc{\int_{B(y^*)} \Zproj(Qy)\,dy}{Z_{J^\perp, h(y^*)}}.
\end{align*}
Hence
\begin{align*}
    Z_1 \prodo \ell{M} \E_{P_\ell} g_\ell 
    &= Z_{O,h(y^*)} \prodo \ell {M-1} \fc{Z_{\be_{\ell+1}J^\perp , h(y^*)}}{Z_{\be_\ell J^\perp , h(y^*)}}\cdot %\fc{\int_{B(y^*)}\Zproj(Qy)\,dy}{Z_{J^\perp , h(y^*)}} = \int_{B(y^*)} \Zproj(Qy)\,dy.
    \fc{\int_{B(y^*)} \Zproj(Qy)\,dy}{Z_{J^\perp, h(y^*)}} = \int_{B(y^*)} \Zproj(Qy)\,dy.
\end{align*}

\paragraph{Variance of $g_\ell$.} With $g_\ell(\si) = \exp(\rc 2(\be_{\ell+1}-\be_\ell) \si^\top J^\perp  \si) = \exp\pa{\rc{2n}\si^\top J^\perp \si}$, we bound
\begin{align}
g_\ell(\si) &\le %\exp\pa{\rc 2(\be_{\ell+1}-\be_\ell) \si^\top J^\perp  \si} \le 
\exp\pa{\rc{2n} \cdot n} = e^{1/2},& 
\fc{\E_{P_\ell} g_\ell^2}{(\E_{P_\ell}g_\ell)^2} 
    &\le
    \E_{P_\ell} g_\ell^2 \le e. 
    \label{e:g2}
\end{align}
We also need to check the variance of 
\begin{align*}
g_{M}(\si) &= \exp\pa{-\fc n2 \ve{y^*}^2}
        %\fc{\exp\pa{\rc 2 \an{\si, J_-\si}} 
        %}{\exp\pa{\an{\mu(y^*),\si}}}
        \int_{B(y^*)} \exp\pa{\an{X^\top Q (y-y^*), \si} + \fc n2 \pa{\ve{y^*}^2-\ve{y}^2}}\,dy
% \fc{\int_{B(y^*)}
%         \exp\pa{\sumo in \pa{-\rc 2 \ve{\si_i P^\prl x_i - Qy}^2 + h_i\si_i}}\,dy,
%         }{\exp\pa{h(y^*)^\top \si}}\\
% &= \exp\pa{-\rc 2 \ve{P^\prl X}_F^2 - \fc n2 \ve{Qy^*}^2}\cdot 
% \fc{\int_{B(y^*)}
%         \exp\pa{\sumo in \pa{-\rc 2 \ve{\si_i P^\prl x_i - Qy}^2 + h_i\si_i}}\,dy
%         }{
%         \exp\pa{\sumo in \pa{-\rc 2 \ve{\si_i P^\prl x_i - Qy^*}^2 + h_i\si_i}}
%         }\\
% &= \begin{multlined}[t] \exp\pa{-\rc 2 \ve{P^\prl X}_F^2 - \fc n2 \ve{Qy^*}^2}\\
% \cdot \int_{B(y^*)}
% \exp\pa{\sumo in \pa{-\rc 2 \ve{\si_i P^\prl x_i - Qy}^2 +\rc 2 \ve{\si_i P^\prl x_i - Qy^*}^2}}\,dy
%\end{multlined}
\end{align*}
%Note that by construction, $\ve{\fc{X^\top P^\prl X}{n}}_\op\le 1$, so $\ve{P^\prl X}_\op\le \sqrt n$.
Note that $\ve{\fc{X^\top P^\prl X}{n}}_\op\le \ve{J}_\op$, so $\ve{P^\prl X}_\op\le \sqrt {n\ve{J}_\op}$.
We check how much the exponent can vary on $B(y^*)$:
\begin{align}
\nonumber
\ab{\an{X^\top Q (y-y^*), \si} + \fc n2 \pa{\ve{y^*}^2-\ve{y}^2}}
&\le \ve{y-y^*} \pa{\ve{P^\prl X\si} + \fc{n}{2}\ve{y^*+y}}\\
% \ab{\fc 12 \sumo in 
% \pa{
% \ve{\si_i P^\prl x_i - Qy^*}^2
% -
% \ve{\si_i P^\prl x_i - Qy}^2
% }}
% &=
% \ab{\fc n2 \pa{\ve{Qy^*}^2 - \ve{Qy}^2} + \an{Q(y-y^*), \sumo in \si_i P^\prl x_i}}\\
% \nonumber
% &\le \fc n2
% \ab{\an{y^*-y,y^*+y}}
% + \eta \sqrt d \cdot \ve{P^\prl X}_2\sqrt n\\
% &\le \fc n2 \cdot \fc\eta2 \sqrt d\cdot 2 L\sqrt d + \eta n\sqrt d \le \rc 2
&\le \fc{\eta}2 \sqrt d\pa{\sqrt {n\ve{J}_\op} \sqrt n + nL\sqrt d} \le \rc 2
\label{e:vary-on-box}
\end{align}
when $\eta\le \rc{ndL + 2n\sqrt{\ve{J}_\op d}}$. This makes $\fc{\E_{p_{M}} [g_{M}(\si)^2]}{\E_{p_{M}} g_{M}(\si)^2}\le e$ as well. We note $g_{M}$ can be easily evaluated since it can be written as a product of integrals of a Gaussian on an interval.

\paragraph{Bias of $\E g_\ell$.}
For the approximate sampling oracle, we let $\wt p_\ell$ be the distribution after running Glauber dynamics for $\Te\pa{n^2 \log  \pf{n}{\ep}}$ steps (for an appropriate choice of constant). Then by Theorem~\ref{l:ls-ising} and~\eqref{e:g2}, $\ab{\E_{P_\ell}g_\ell(\si) - \E_{\wt P_\ell}g_\ell(\si)}\le
d_{\TV} (P_\ell, \wt P_\ell) \cdot e^{1/2} \le \fc{\ep}{4M}$.

\paragraph{Using Lemma~\ref{l:sa}.}
%Let \[Z(y) = \exp\pa{\rc 2 \an{\si, J^\perp \si}+\an{h,\si} -\rc 2\sumo in  \ve{\si_i P^\prl x_i - Qy}^2 }.\] 
By Lemma~\ref{l:sa} with $\de$ replaced by $\fc{\de}{(L/\eta)^d}$, using a union bound, we obtain that with probability $\ge 1-\de$, for all $y^*\in \Grid$, 
$\wh Z(y^*) \in [e^{-\eph}, e^{\eph}] \cdot \int_{B(y^*)} \Zproj(Qy)\,dy$ and so 
\begin{align*}
     \sum_{y^*  \in \Grid} \wh Z(y^*)
     \in  [e^{-\eph}, e^{\eph}] \cdot \sum_{y^*\in \Grid} \int_{B(y^*)} \Zproj(Qy)\,dy
     = [e^{-\eph}, e^{\eph}] \cdot\int_{\ve{y}_\iy\le L} \Zproj(Qy)\,dy.
\end{align*}

\paragraph{Error from cutoff.} We would like to estimate
$\Zproj  = \int_{\R^d}\Zproj(Qy)\,dy$, so it remains to show that at least $e^{-\eph}$ of the probability mass of $p(\si,y)$ is contained in $\{\pm 1\}^n \times [-L,L]^d$. For this, it suffices to fix $\si$, and show that $P(y \nin [-L,L]^d|\si) \le \eph $. We have by Lemma~\ref{l:joint-prl}(3) that
\[
\pproj(y |\si) = \pf{n}{2\pi}^{d/2} \exp\pa{-\fc n2 \ve{Qy - \fc{\sumo in \si_i P^\prl x_i}n}^2}.
\]
%Note that by construction, $\ve{\fc{X^\top P^\prl X}{n}}_2\le 1$, so $\ve{X^\top P^\prl}_2\le \sqrt n$ and 
Using $\ve{X^\top P^\prl}_2\le \sqrt {n\ve{J}_\op}$, we get
\[
\ve{\fc{\sumo in \si_i P^\prl x_i}n} \le \ve{\fc{X^\top P^\prl}{n}}_2 \sqrt {n} \le \sqrt{\ve{J}_\op}.
\]
Hence taking 
$L = \Om \pa{\sqrt{\ve{J}_\op} + 1} = \Om\pa{\sqrt{\ve{J}_\op} + \sqrt{ \log  \pf n\ep/n}}$, 
we have
\begin{align}\label{e:cutoff}
P(y \nin [-L,L]^d)& \le \sumo in P(y_i\nin [-L,L]) = n\cdot \fc{\ep}{2n} = \fc{\ep}2.
\end{align}
Putting everything together and using %Corollary~\ref{c:Z}, 
Theorem~\ref{t:hs}, 
we have with probability $\ge 1-\de$ that 
\[
%\exp\pa{-\rc 2 \ve{P^\prl X}_F^2} 
\pf{n}{2\pi}^{\fc d2} %\exp\pa{\fc n2 \Tr(J^\prl)}
\sum_{y^* \in \Grid} \wh Z(y^*)
    \in [e^{-\ep}, e^{\eph}]\cdot  \Zproj .
\]
\end{proof}

\subsection{Estimation with positive and negative spikes}
%Partition function estimation with negative and positive spikes}

We now analyze Algorithm~\ref{a:Z-ising} when there are negative spikes to prove Theorem~\ref{t:main}(1).

%We now describe how to modify Algorithm~\ref{a:Z-ising} to obtain Theorem~\ref{t:main}(1) when there are negative spikes. Let $V$ denote the subspace of $\R^n$ spanned by eigenvectors of $J$ with eigenvalues $>1-\rc c$ and $P^\perp$ projections onto $V^\perp$, and $J'=\fc{X^\top P^\perp X}{n}$. Write $J' = J_+-J_-$, where $J_+$ and $J_-$ are non-negative definite and negative definite, respectively, whose column spaces intersect only in 0. 

For $y^*\in \Grid$, Algorithm~\ref{a:Z-ising} uses Corollary~\ref{thm:negdef-ising} to find 
$\mu(y^*)$ such that 
%an approximate critical point $\mu(y^*)$ of $\log \E_{p_{J_+,X^\top Q y^*}} [e^{\an{u,J_+\si}}] - \rc 2\an{u,Ju}$. 
%Note that although Theorem~\ref{thm:sgd}
%Then it is the case that 
\begin{align}\label{e:dd-ub}
\log \pa{\dd{P_{J_+,X^\top Qy^*+h}}{P_{J,X^\top Qy^*+h+\mu(y^*)}}} \le c\Tr(J_-) +1.
\end{align}
%(Note that although Theorem~\ref{thm:negdef} asks for an exact sampling oracle, if we ensure that the total TV-errors over all calls to the sampler is bounded by $\de$, then we can use a coupling argument to restrict to an event of probability $\ge 1-\de$ where~\eqref{e:dd-ub} holds.) 
Let $J^\perp_{\textup{all}} = J^\perp - J_- = J-J^\prl$.
We first calculate
\begin{align}
\nonumber
    \E_{P_{M}}g_{M} &= 
    \rc{Z_{J^\perp, h(y^*)}} \sBC \ising{J^\perp}{h(y^*)}
    \fc{\exp\pa{-\rc 2\an{\si, J_-\si}}}{\exp\pa{\an{\mu(y^*),\si}}}\\
    \nonumber
    &\quad \cdot \int_{B(y^*)} \exp\pa{\an{X^\top Q (y-y^*), \si} 
        -\fc n2 \ve{y}^2} \,dy\\
    &=\rc{Z_{J^\perp, h(y^*)}} \sBC \int_{B(y^*)} \exp\pa{\rc 2 \an{\si, (J^\perp_{\textup{all}} - J_-)\si} + \an{X^\top Q y+h, \si} 
        -\fc n2 \ve{y}^2} \,dy\label{e:pgn1}\\
    \nonumber
    &=  \fc{\int_{B(y^*)} \Zproj(Qy)\,dy}{Z_{J^\perp, h(y^*)}}
\end{align}
as before.

% In Algorithm~\ref{a:Z-ising}, we sample from $p_{\fc{\ell-1}n J_+,h(y^*)+\mu(y^*)}$ instead. Finally, we let
% \begin{align*}
% g_{n+1}(\si) 
% &= \fc{\exp\pa{-\rc 2 \si^\top J_-\si}  \int_{B(y^*)}
%         \exp\pa{\sumo in \pa{-\rc 2 \ve{\si_i P^\prl x_i - Qy^*}^2 + h_i\si_i}}\,dy}{\exp\pa{(h(y^*) + \mu(y^*))^\top \si}}
% \end{align*}

We now bound $\fc{\E g_{M}^2}{(\E g_{M})^2}$. 
First we bound 
\begin{align*}
g_{M}(\si)&= 
\fc{\exp\pa{-\rc 2 \an{\si, J_-\si}} 
        }{\exp\pa{\an{\mu(y^*),\si}}}\int_{B(y^*)} \exp\pa{\an{X^\top Q (y-y^*), \si} 
        -\fc n2 \ve{y}^2} \,dy\\
&=\fc{\exp\pa{\rc 2\an{\si,J^\perp_{\textup{all}} \si} + \an{X^\top Q y^*+h, \si}}}{\exp\pa{\rc 2 \an{\si, J^\perp\si} + \an{h(y^*),\si}}} \cdot 
\int_{B(y^*)} \exp\pa{\an{X^\top Q (y-y^*), \si} 
        -\fc n2 \ve{y}^2} \,dy\\
&\le \exp\pa{c\Tr(J_-)+1} \fc{Z_{J^\perp_{\textup{all}},X^\top Q y^*+h}}{Z_{J^\perp,h(y^*)}} \exp\pa{-\fc n2 \ve{y^*}^2}\eta^d e^{1/2}
%         \exp\pa{\sumo in \pa{-\rc 2 \ve{\si_i P^\prl x_i - Qy^*}^2 + h_i\si_i}}\,dy}{\exp(h(y^*)^\top \si)} \\
% &\le \exp\pf{c\Tr(J_-)}2 \fc{Z_{J,h(y^*)}}{Z_{J_+,h(y^*)+\mu(y^*)}} \eta^d e^{1/2}
\end{align*}
using~\eqref{e:dd-ub} and~\eqref{e:vary-on-box}. 
Next, again using~\eqref{e:vary-on-box}, we bound
\begin{align*}
\E_{P_{M}} g_{M}(\si) &= \sBC
\fc{\exp\pa{\rc 2\an{\si,J^\perp_{\textup{all}} \si} + \an{X^\top Q y^*+h, \si}}}{Z_{J^\perp,h(y^*)}}\\
&\quad \cdot 
\int_{B(y^*)} \exp\pa{\an{X^\top Q (y-y^*), \si} 
        -\fc n2 \ve{y}^2} \,dy\\
&\ge \fc{Z_{J^\perp_{\textup{all}},X^\top Q y^*+h}}{Z_{J^\perp,h(y^*)}} \exp\pa{-\fc n2 \ve{y^*}^2}\eta^d e^{-1/2}
\end{align*}
Hence 
\begin{align*}
    \fc{\E_{P_{M}} g_{M}^2}{(\E_{P_{M}} g_{M})^2}
    &\le \exp(2c\Tr(J_-)+4),
\end{align*}
and this is the extra multiplicative error we incur in estimation. The rest of the estimates in the proof are the same as before.

%\hlnote{Add note on error in using SGD (coupling argument).}

The above concludes the proof of our main result for computing the partition function. We now briefly discuss the performance of this algorithm under the ``naive mean field'' assumption $\|J\|_F^2 = o(n)$ referenced in the introduction and introduced in \citep{basak2017universality}.
\begin{rem}\label{rem:speed-mf}
Suppose we want to bound the performance of the algorithm from Theorem~\ref{t:main} in terms of Frobenius norms. This will be very wasteful compared to the original statement, but is useful for comparison. 

For simplicity, we can make the common assumption that the diagonal of $J$ is zero, which means that the sum of the eigenvalues of $J$ is zero. Then we can choose the interval $[-1/3,1/3]$ as the interval of length at most one in the application of the Theorem. %Decompose $J$ into psd and negative definite parts as $J = J_+ - J_-$. 
The runtime for estimating $\log Z$ to additive $\ep$ error will be at most
\[ O\pf{(\ve{J}_\op n)^{O(d_++1)} e^{O(\la_1+\cdots + \la_{d_-} - \ d_-/3)}}{\ep^2} \]
where $-\lambda_1,\ldots, -\la_{d_-}$ are the eigenvalues of $J$ below $-1/3$. Now clearly we have $\sum_{i = 1}^{d_-} \lambda_i \le \sum_{i = 1}^{d_-} 3\lambda_i^2 \le 3\|J\|_F^2$ and $d_+ \le 3\|J\|_F^2$. So we have a crude bound on the runtime as
\[ O\left( \frac{(n\|J\|_{\op})^{O(\|J\|_F^2})}{\ep^2} \right). \]
In particular, provided $\|J\|_F^2 = o(n/\log(n))$ we have that this is subexponential time. So the result works up to almost the same subexponential time regime as the algorithm in the work \citep{jain2019mean} when specialized to the setting of Ising models. Depending on the precise properties of $J$, the precise runtime of the new algorithm could be faster or slower than the algorithm of \citep{jain2019mean}, but the approximation error for this one is much stronger (additive error $\ep$ to $\log Z$).
\end{rem}

%% file: sampling.tex
\section{Sampling}\label{s:sampling}
We now turn to the problem of generating samples from the model; for the reader, note that this section builds on results and uses notation from the previous section on partition function estimation.

By choosing $y^*\in \Grid = \bc{-L + \rc 2 \eta, -L + \fc 32 \eta, \ldots, L-\rc 2\eta}^d$ with probability proportional to $\wh Z(y^*)$ estimated by Algorithm~\ref{a:Z-ising} and then sampling from $\Pproj^{\si, y}$ restricted to $\BC\times B(y^*)$, we can obtain an algorithm for sampling of the same order of complexity as in Theorem~\ref{t:Z}. %, $O\pa{\sqrt n \log\pf n\ep}^{\rank_1(J)} \cdot O\pf{\poly(n) \log\prc{\ep\de}}{\ep^2}$. 
In this section, we give an algorithm that only has logarithmic dependence on $\ep$ and prove Theorem~\ref{t:main}(2).  %$\log\prc{\ep}$ dependence. 

% \begin{thm}\label{t:samp}
% Suppose $J$ is PSD. 
% %There is a algorithm that
% Algorithm~\ref{a:st} outputs a sample from a distribution that is $\ep$-close in TV-distance to $p_{J,h}$ 
% in expected time 
% $O\pa{n \ve{J}_\op \log\pf 1\ep}^{O(1+\rank_1(J))}$.
% %$\ve{J}_2n^{O(1+\rank_1(J))}\log\prc\ep$.
% \end{thm}

Let $Z_\ell(y^*):= Z_{\beta_\ell J^{\perp}, h(y^*)}$ for $1\le \ell\le M$ and $Z_{M+1}(y^*):= \int_{B(y^*)} \Zproj(Qy)\,dy$. Denote the approximately normalized probabilities
\begin{align*}
    q_{\ell, y^*}
    & = %\pf{\eta}{L}^n 
    \fc{Z_{\ell}(y^*)}{\wh Z_{\ell}(y^*)}  
    p_{\be_\ell J^{\perp},h(y^*)}
\end{align*}
where $\be_{\ell} = \fc{\ell-1}n$. Overloading notation, we will also write $p_{\ell,y^*}$ for $p_{\be_\ell J^{\perp},h(y^*)}$. Note that we can compute the ratios of different $q_{\ell,y^*}$'s, as we have $q_{\ell,y^*}(\si)\propto \rc{\wh Z_{\ell}(y^*)} \ising{\be_\ell J^\perp}{h(y^*)}$. 
%\hlnote{Need to define $Z_{\ell,y^*}$, and say that Algorithm~\ref{a:sa} also estimates them.}

%Let $\Grid =  \bc{-L + \rc 2 \eta, -L + \fc 32 \eta, \ldots, L-\rc 2\eta}^d$. 
We define a Markov chain on an expanded state space $\{1,\ldots, M\} \times \Grid \times \{\pm 1\}^n$, where the first index denotes the ``temperature" of the distribution. This is similar to a simulated tempering chain~\citep{marinari1992simulated}, with two types of moves: between temperatures and within temperatures.
However, there are two differences with a standard simulated tempering chain:
\begin{enumerate}
    \item We use a different normalizing constant $\wh Z_{\ell}(y^*)$ for each value of $y^*$, in order to make sure the stationary distribution is roughly uniformly distributed over the $y^*\in \Grid$.
    \item Within any temperature other than the highest one, we do not allow moves that change $y^*$.
\end{enumerate}
% However, in contrast to a standard simulated tempering chain, we use a different normalizing constant $Z_{\ell, y^*}$ for each value of $y^*$, in order to make sure the stationary distribution is roughly uniformly distributed over the boxes.
Finally, we do simulated tempering on the space $\Grid\times \{\pm 1\}^n$ rather than $[-L,L]^d\times \{\pm 1\}^n$ for convenience; this adds an extra rejection sampling step at the end where we compare the distributions on $\{y^*\}\times \{\pm 1\}^n$ and on $B(y^*)\times \{\pm 1\}^n$, similar to the final ratio $g_{M}$ in partition function estimation. 

We need the modifications for technical reasons to make our proof work; it is an interesting question whether a more standard simulated tempering chain would work. Our proof strategy is based on a Markov chain decomposition theorem similar to~\cite{ge2018simulated}, which we will now introduce.

%We define the Markov chain on $\{0,\ldots, M\} \times \bc{-L + \rc 2 \eta, -L + \fc 32 \eta, \ldots, L-\rc 2\eta}^d \times \{\pm 1\}^n$ as follows. 

 \begin{algorithm2e}[h!]
 \DontPrintSemicolon
 \caption{Simulated tempering on $\{1,\ldots, M+1\} \times \Grid \times \{\pm 1\}^n$} 
 \KwIn{Ising model $(J,h)$, steps to run Markov chain $T$ (suggested $\Te(n^4d\log(n\ve{J}_\op/\ep))$.}
 %\Require Starting point (or distribution) $(\ell, y^*, \si)$. 
 %\Require Query access to probability distribution $p(\si)$ on $\{\pm 1\}^n$, up to constant of proportionality; number of steps $T$.
 %\Ensure Approximate sample from $p$.
 %Approximation of partition function $Z_{J,h}$.
 Run Algorithm~\ref{a:Z-ising}  
 to obtain partition function estimates $\wh Z_\ell (y^*)$ for $1<\ell \le M+1=n+2$.\;
  Let $\ell=1$. Draw $y^*\in \Grid= \bc{-L + \rc 2 \eta, -L + \fc 32 \eta, \ldots, L-\rc 2\eta}^d$, and then draw $\si\sim P^{\si|y}_{O,h(y^*)}(\cdot|y^*)$. \;
 \For{$1\le t\le T$} {
     With probability $\rc 4$, if $\ell \ne M$, set $\ell\leftarrow \ell+1$ with probability $\min\bc{\fc{q_{\ell+1,y^*}(\si)}{q_{\ell,y^*}(\si)},1}$.\;
     With probability $\rc 4$, if $\ell \ne 1$, set $\ell\leftarrow \ell-1$  with probability $\min\bc{\fc{q_{\ell-1,y^*}(\si)}{q_{\ell,y^*}(\si)},1}$.\;
    %\Block{With probability $\rc 2$,} 
    With probability $\rc 2$,
    \Begin{
        \If{$\ell=1$}{
             With probability $\rc 2$, reselect a random $y^*\in \Grid$, and then draw $\si\sim P^{\si|y}_{O,h(y^*)}(\cdot|y^*)$. 
        }
        Choose a random coordinate $i$, and set $\si \leftarrow \si^{(i)}$ with probability $\fc{q_{\ell, y^*}(\si^{(i)})}{q_{\ell,y^*}(\si^{(i)})+q_{\ell,y^*}(\si)}$.\;
        % \eIf{$\ell=1$}{
        %      Reselect a random $y^*\in \Grid$, and then draw $\si\sim P^{\si|y}_{O,h(y^*)}(\cdot|y^*)$. 
        % }{
        %      Choose a random coordinate $i$, and set $\si \leftarrow \si^{(i)}$ with probability $\fc{q_{\ell, y^*}(\si^{(i)})}{q_{\ell,y^*}(\si^{(i)})+q_{\ell,y^*}(\si)}$.
        % }
    }
 }
 \If{$\ell=M$}{
     Draw $U\sim \mathsf{Uniform}([0,1])$.\;
    \If{{$U\le 
    (4e \max\wh Z_{M+1}(y^*)\exp(c\Tr(J_-)+1))^{-1} \wh Z_M(y^*) g_{n+1,y^*}(\si)$
    %e^{-1/2} 
    %\fc{\wh Z_{M+1}(y^*) \exp\pa{-\rc 2 %\ve{y^*}^2}}{e^{1/2}\max_{y\in \Grid} \wh Z_{M+1} (y)} \pf{1}{\eta}^d \int_{B(y^*)}
%\exp\pa{\sumo in \pa{-\rc 2 \ve{\si_i P^\prl x_i - Qy}^2 +\rc 2 \ve{\si_i P^\prl x_i - Qy^*}^2}}\,dy$
    }}{
         Return $\si$.
    }
 }
 If failed to return sample, re-run the procedure.
  \label{a:st}
 \end{algorithm2e}

Given a Markov chain on $\Om$, we define two Markov chains associated with a partition of $\Om$.
\begin{df}[\cite{madras2002markov}]\label{df:assoc-mc}
For a Markov chain $\cal M=(\Om, T)$, and a set $A\subeq \Om$, define the \vocab{restriction of $\cal M$ to $A$} to be the Markov chain $\cal M|_A = (A, T|_{A})$, where
$$
T|_A(x,B) = T(x,B) + \one_B(x) T(x,A^c).
$$
(In words, $T(x,y)$ proposes a transition, and the transition is rejected if it would leave $A$.)

Suppose the unique stationary measure of $\cal M$ is $P$. 
Given a partition $\cal P = \set{A_j}{j\in J}$, define the \vocab{projected Markov chain with respect to $\cal P$} to be $\ol{\cal M}^{\cal P} = (J, \ol T^{\cal P})$, where
$$
\ol T^{\cal P} (i,j) = 
%\rc{\mu(A_i)} \int_{A_i}\int_{A_j} P(x,dy)\mu(dx).
\rc{P(A_i)} \int_{A_i}\int_{A_j} T(x,dy)\,P(dx).
$$
(In words, $\ol T(i,j)$ is the ``total probability flow'' from $A_i$ to $A_j$.)
We omit the superscript $\cal P$ when it is clear.
\end{df}

The following theorem lower-bounds the gap of the original chain in terms of the gap of the projected chain and the minimum gap of the restricted chains.
\begin{thm}[\cite{madras2002markov}]
\label{t:decomp}
Let $\cal M=(\Om, T)$ be a Markov chain with stationary measure $P$. Let $\cal P=\set{A_j}{j\in J}$ be a partition of $\Om$ such that $P(A_j)>0$ for all $j\in J$. Then 
\begin{align*}
    \rc 2 \Gap( \ol{\cal M}^{\cal P}) \min_{j\in J} \Gap (\cal M|_{A_j}) &\le \Gap(\cal M) \le \Gap(\ol{\cal M}^{\cal P}).
\end{align*}
\end{thm}
%\hlnote{Notation: Fix collision between $M$'s. Rewrite Markov chain as $\cal M$.}
We can now prove our main theorem for sampling.

\begin{proof}[Proof of Theorem~\ref{t:main}(2)] %{samp}
Let $\cal M$ be the simulated chain in Algorithm~\ref{a:st}. 
Below, we condition on the event that all the $\wh Z_{\ell}(y^*)$ are $2$-multiplicative approximations of $Z_\ell(y^*)$, that is, $\wh Z_{\ell}(y^*)\in %[1-O(\ep), 1+O(\ep)]
[\rc 2, 2]\cdot Z_{\ell}(y^*)$. As in the proof of Theorem~\ref{t:Z}, if we choose the failure probability to be $O\pa{\fc \ep M \pf{\eta}{2L}^d}$, by Lemma~\ref{l:sa} and a union bound---this time applied to the estimates at all levels $\wh Z_\ell(y^*)$---this event happens with probability $1-O(\ep)$. 

We let $\Pst$ denote the stationary measure for the simulated tempering chain, and $\Pst_\ell$ denote the measure restricted to $\{\ell\}\times \Grid \times \BC$.

%Note 1-\ep of prob mass...

We use Theorem~\ref{t:decomp} with the partition given by $A_{\ell,y^*}= \{\ell\}\times \{y^*\} \times \{\pm 1\}^n$. 
% \[
% d_{\TV} \pa{p|_{\{\mu\in [-L,L]^n\}}, 
% \sumz \ell M\sum_{i\in \{-\fc L\eta, \ldots, \fc L\eta - 1\}^n} \fc{\wt Z_{\ell,i}}{Z_{\ell,i}} p\one_{\{\ell\}\times B_i\times \{\pm 1\}^n}
% } = O(\ep). 
% \]
The restriction $\cal M|_{A_{\ell,y^*}}$ is a 
lazy version of the Glauber dynamics chain for %$p_{\be_\ell J^{\perp}, h(y^*)}$ 
$P_{\ell,y^*}$ (that is, with all transition probabilities halved, or multiplied by $\rc 4$ in the case $\ell=1$), which has Poincar\'e constant bounded by $O(n^2)$ by Lemma~\ref{l:ls-ising}. 
%lazy chain we defined on $q_{\ell,i}$, which has Poincar\'e constant bounded by $O(n)$ by Lemma~\ref{l:ls-slice}. 

First, note that by construction with the Metropolis-Hastings acceptance ratio, the stationary distribution satisfies
\begin{align}\label{e:st-stationary} 
\ol p((\ell,y^*)) \propto R_\ell(y^*):= \fc{Z_\ell(y^*)}{\wh Z_\ell (y^*)}.
\end{align}
For the projected chain, we use Lemma~\ref{l:proj}. We check each of the conditions.
\begin{enumerate}
    \item  To bound the ``bottleneck ratio", note that for
    $k<\ell$, letting %$R_j(y^*) = \fc{Z_j(y^*)}{\wh Z_j(y^*)}$, 
    $\ol p_j(y^*) = \ol p(y^*|j) = \fc{\ol p((j,y^*))}{\sum_{y\in \Grid} \ol p((j,y))}$
    \begin{align*}
        \fc{\ol p_k(y^*)}{\ol p_\ell(y^*)} &= \fc{R_k(y^*)/\sum_{y\in \Grid}R_k(y)}{R_\ell(y^*)/\sum_{y\in \Grid}R_\ell(y)} \ge \rc 4
    \end{align*}
    using the fact that the $\wh Z_j(y)$ are 2-multiplicative approximations, so that $R_j(y^*)\in [\rc 2, 2]$ for each $j,y^*$.
    \item %By Lemma~\ref{l:g-var},
    From~\eqref{e:st-stationary}, %and the fact that $\wh Z_\ell(y^*)$ are 2-multiplicative approximations, 
    we have $\fc{\ol p((\ell,y^*))}{\ol p((\ell,y^*))} \in [\fc 14, 4]$.
    Note that for $\ell, \ell\pm 1\in [M]$, 
    \begin{align*}
        \fc{p_{\ell\pm 1,y^*}(\si)}{p_{\ell,y^*}(\si)} = \exp(\an{\si,(\be_{\ell\pm 1}-\be_\ell)J\si}) \fc{Z_{\ell}(y^*)}{Z_{\ell\pm 1}(y^*)} = \Te(1)
    \end{align*}
    because the ratio of individual terms in $Z_{\ell,y^*}$ and $Z_{\ell\pm 1,y^*}$ is $\Te(1)$. 
    %By~\eqref{e:g2}, for $1\le \ell\le M-1$,
    %$\chi^2(p_{\ell+1,y^*}||p_{\ell,y^*}) = \fc{\E_{p_\ell} g_\ell^2}{(\E_{p_\ell}g_\ell)^2} - 1 = O(1)$ 
    %so 
    Hence
    \begin{align*}
        \ol T((\ell,y^*), (\ell\pm 1,y^*))
        & = \sum_{\si\in \{\pm 1\}^n} 
        \min\bc{\fc{\wh Z_\ell(y^*)/Z_\ell(y^*)}{\wh Z_{\ell\pm 1}(y^*)/Z_{\ell\pm 1}(y^*)} \cdot \fc{p_{\ell\pm 1,y^*}(\si)}{p_{\ell,y^*}(\si)},1} p_{\ell,y^*}(\si)\\
        &\ge \rc 4 \sum_{\si\in \{\pm 1\}^n}
        \min\bc{\fc{p_{\ell\pm 1,y^*}(\si)}{p_{\ell,y^*}(\si)},1}p_{\ell,y^*}(\si)
        %\min\bc{p_{\ell\pm 1,y^*}(\si), p_{\ell,y^*}(\si)}
        %= \rc 4 (1-d_{\TV}(p_{\ell\pm 1,y^*}(\si), p_{\ell,y^*}(\si)))
        %&=\Om\pa{\min\bc{\fc{P((\ell\pm 1, i))}{P((\ell, i))},1}}\\
        %P((0,i), (0,j)) &= \Om(q_0(B(j)).
        \\
        &=\Om(1) = \Om\pf{\ol p((\ell\pm 1,y^*))}{\ol p((\ell,y^*))}
    \end{align*}
    where we used the fact that $\wh Z_\ell(y^*)$ are 2-multiplicative approximations. We also note
    \begin{align*}
        \ol T((1,y^*), (1,z^*)) &\ge \rc{4}\pf{\eta}{2L}^d.
    \end{align*}
    Hence, condition 1 of Lemma~\ref{l:proj} holds with constant $D_{\text{high}}$ and $D_{\text{adj}}$.
    % \item To bound the ``bottleneck ratio", note that for
    % $k<\ell$, letting $R_j(y^*) = \fc{Z_j(y^*)}{\wh Z_j(y^*)}$, 
    % \begin{align*}
    %     \fc{P_k(j)}{P_\ell(j)} &= \fc{R_k(y^*)/\sum_{y\in \Grid}R_k(y)}{R_\ell(y^*)/\sum_{y\in \Grid}R_\ell(y)} \ge \rc 4
    % \end{align*}
    % again using the fact that the $\wh Z_j(y^)$ are 2-multiplicative approximations, so that $R_j(y^*)\in [\rc 2, 2]$ for each $j,y^*$.
    %$\fc{q_k(B(i))}{q_\ell(B(i))} = \Om\pa{\pf \eta L^d}$. 
    \item Finally, for any $1\le \ell\le M$, 
    \[
    P\pa{\{\ell\} \times\Grid \times \{\pm 1\}^n} =
    \fc{\sum_{y\in \Grid} R_\ell(y)}{\sumo \ell{M}\sum_{y\in \Grid} R_\ell(y)}
    \ge \rc{4M}.
    \]
\end{enumerate}
Hence by Lemma~\ref{l:proj}, the Poincar\'e constant of $\cal M$ is $O(M^2)=O(n^2)$. Since $\cal M|_{A_{\ell,y^*}}$ have Poincar\'e constant bounded by $O(n^2)$ for each $y^*$,  noting the spectral gap is the inverse of the Poincar\'e constant and using Lemma~\ref{t:decomp}, we get that the Poincar\'e constant of $\cal M$ is $\CP= O\pa{n^2\cdot  n^2} = O\pa{n^4}$. % This gives the desired bound.
 %Taking $L=\sqrt{n\ve{J}_2} + \Te\pa{\sqrt n \log \pf n\ep}$ and $\eta = \Te \prc{nLr + n\ve{J}_2\sqrt r}$, the total mixing time is $\pa{n\pa{\ve{J}_2 + \log \pf 1\ep}}^{O(d)}\log \prc \ep$.
% \fixme{Mixing time}
For the mixing time, note that the starting distribution is the restriction of the stationary distribution to $\{1\}\times \Grid \times \{\pm 1\}^n$, which has at least $\rc{4M}$ of the mass. Hence the time until the distribution is $\ep$-close to the stationary distribution (and all restrictions to $\{\ell\}\times \Grid \times \{\pm 1\}^n$, $1\le\ell\le M$ are $\ep$-close) is $O\pa{\CP\log\pf{M}{\ep}}$.

Let $P_{M+1}$ be the probability measure on $\{\pm 1\}^n\times \Grid$ with probability mass function given by 
\[
p_{M+1}(\si,y^*) = \fc{\int_{B(y^*)} \exp\pa{\rc 2 \an{\si, \Jpall \si} + \an{X^\top Qy + h, \si} - \fc n 2\ve{y}^2}\,dy}{\int_{[-L,L]^d}\sum_{\si\in \{\pm 1\}^d} \exp\pa{\rc 2 \an{\si, \Jpall \si} + \an{X^\top Qy + h, \si} - \fc n2\ve{y}^2}\,dy}
%=\fc{\int_{B(y^*)} \exp\pa{\rc 2 \si^\top J^{\perp} \si + h(y)^\top \si - \rc 2\ve{y}^2}\,dy}{\sum_{\si\in \{\pm 1\}^d}\int_{[-L,L]^d} \exp\pa{\rc 2 \si^\top J^{\perp} \si + h(y)^\top \si - \rc 2\ve{y}^2}\,dy}
;
\]
that is, it is obtained from restricting $p_{\Jpall, X^\top Qy + h}^{\si,y}(\si,y)$ to $\{\pm 1\}^n \times [-L,L]^d$ and then rounding $y$ to the nearest grid point. Except for the fact that this measure is restricted to $[-L,L]^d$, this is the distribution we wish to sample from. 
We also know that 
\[
\pst_M(\si,y) = \pa{\wh Z_M(y^*) \sum_{y\in \Grid} R_M(y)}^{-1} \ising{J^\perp}{h(y^*)}.
\]
In terms of $\fc{p_{M+1}(\si,y)}{\pst_M(\si,y)}$, 
the acceptance ratio in Algorithm~\ref{a:st} is given by
\begin{multline}
(4e \max_{y^*\in \Grid}\wh Z_{M+1}(y^*)\exp(c\Tr(J_-)+1))^{-1} \wh Z_M (y^*) g_{n+1,y^*}(\si)\\
= (4e \max_{y^*\in \Grid}\wh Z_{M+1}(y^*)\exp(c\Tr(J_-)+1))^{-1}  \cdot \fc{\int_{[-L,L]^d} Z_{J^\prl, \Jpall, h}(Qy)\,dy}{\sum_{y\in \Grid} R_M(y)} \cdot \fc{p_{M+1}(\si,y^*)}{\pst_{M}(\si,y^*)}
%    \rc{e^{1/2}\eta^d\max_{y\in \Grid} \wh Z_{M+1}(y)}
%    \fc{\int_{[-L,L]^d}\sum_{\si\in \{\pm 1\}^d} \exp\pa{\rc 2 \si^\top J^{\perp} \si + h(y)^\top \si - \rc 2\ve{y}^2}\,dy}{\sum_{y\in \Grid}\fc{Z(y)}{\wh Z(y)}}
%    \fc{p^*(\si,y)}{p_M(\si,y)}.
\label{e:ar}
\end{multline}
This is a constant times $\fc{p_{M+1}(\si,y^*)}{\pst_{M}(\si,y^*)}$, so it is the correct rejection sampling ratio. 
We need to show that this is always at most 1, and give a lower bound for the coefficient of $\fc{p_{M+1}(\si,y)}{\pst_M(\si,y)}$.
\begin{enumerate}
    \item Ratio is at most 1:
    We first consider
    \begin{align*}
        \fc{p_{M+1}(\si,y^*)}{\pst_M(\si,y^*)} 
        &= \fc{p_{\Jpall, X^\top Qy^*+h}(\si)}{p_{J^\perp, h(y^*)}(\si)  \pst_M(y^*)}
        \cdot \fc{p_{M+1}(\si|y^*)p_{M+1}(y^*)}{p_{\Jpall, X^\top Qy^*+h}(\si)}\\
        & = \fc{p_{\Jpall, X^\top Qy^*+h}(\si)}{p_{J^\perp, h(y^*)}(\si) \fc{R_M(y^*)}{\sum_{y\in \Grid}R_M(y)}} \cdot 
         \fc{p_{M+1}(\si|y^*)}{p_{\Jpall, X^\top Qy^*+h}(\si)} \fc{\int_{B(y^*)} Z_{J^\prl , J^\perp , h}(Qy)\,dy}{\int_{[-L,L]^d} Z_{J^\prl , J^\perp , h}(Qy)\,dy}\\
        &\le 2\cdot \pa{\sum_{y\in \Grid}R_M(y)}
        \exp(c\Tr(J_-)+1) \cdot 
         \fc{p_{M+1}(\si|y^*)}{p_{\Jpall, X^\top Qy^*}(\si)} \fc{\int_{B(y^*)} Z_{J^\prl , J^\perp , h}(Qy)\,dy}{\int_{[-L,L]^d} Z_{J^\prl , J^\perp , h}(Qy)\,dy}
    \end{align*}
    where we used the guarantee obtained from Corollary~\ref{thm:negdef-ising}. 
    We also note 
    \begin{align*}
        \fc{p_{M+1}(\si|y^*)}{p_{\Jpall, X^\top Qy^*+h}(\si)} &\propto \int_{B(y^*)} \exp\pa{\an{X^\top Q(y-y^*)}-\fc n2 \ve{y}^2}\,dy \\
        &\in \exp\pa{-\fc n2 \ve{y^*}^2} \cdot [e^{-1/2},e^{1/2}]
    \end{align*}
    by~\eqref{e:vary-on-box}; hence, because probabilities integrate to 1, the ratio is bounded by $e$.
    Combining with~\eqref{e:ar}, we obtain that the acceptance ratio is bounded by 
    \begin{align*}
        \rc{2} \cdot %\fc{p_{M+1}(\si,y^*)}{\pst_{M}(\si,y^*)}
        \fc{\int_{B(y^*)} \Zproj(Qy)\,dy}{\max_{y^*\in \Grid} \wh Z_{M+1}(y^*)}
        &\le \rc{2} \cdot \fc{\int_{B(y^*)} \Zproj(Qy)\,dy}{\rc 2 \max_{y^*\in \Grid}  Z_{M+1}(y^*)}\le 1.
    \end{align*}

%     We write $g_{M,y^*}(\si)$ as
%     \begin{align*}
%         g_{M,y^*}(\si) &= \fc{Z_{M+1}(y^*)}{Z_{M}(y)} 
%         \fc{p_{M+1}(\si)}{p_{M}(\si)}
%     \end{align*}
%     By~\eqref{e:vary-on-box}, the acceptance ratio is at most
%     \begin{multline*}
%         \fc{\exp\pa{-\rc 2 \ve{y^*}^2} \wh Z(y^*)}{e^{1/2}\eta^d \max_{y\in \Grid} \wh Z(y)}
%         \int_{B(y^*)}
% \exp\pa{\sumo in \pa{-\rc 2 \ve{\si_i P^\prl x_i - Qy}^2 +\rc 2 \ve{\si_i P^\prl x_i - Qy^*}^2}}\,dy\\
%         \le \rc{\eta^d e^{1/2}} e^{1/2} \Vol(B(y^*)) = 1.
%     \end{multline*}
    \item Lower bound for coefficient: The reciprocal of the coefficient is
    \begin{align*}
    &4e \fc{\max\wh Z_{M+1}(y^*)}{\int_{[-L,L]^d} Z_{J^\prl, \Jpall, h}(Qy)\,dy} \exp(c\Tr(J_-)+1)
\cdot {\sum_{y\in \Grid} R_M(y)}\\
    &\le 4e \cdot  \exp(c\Tr(J_-)+1)\cdot 2 \max_{y^*\in \Grid} \fc{\int_{B(y^*)}Z_{J^\prl, \Jpall, h}(Qy)\,dy}{\int_{[-L,L]^d} Z_{J^\prl, \Jpall, h}(Qy)\,dy} \cdot 2\pf{2L}{\eta}^d\\
    &\le 16  \exp(c\Tr(J_-)+2) \pf{2L}{\eta}^d.
    % &\fc{e^{1/2}\eta^d\max_{y\in \Grid} \wh Z_{M+1}(y)\sum_{y\in \Grid}\fc{Z(y)}{\wh Z(y)}}{\int_{[-L,L]^d}\sum_{\si\in \{\pm 1\}^d} \exp\pa{\rc 2 \si^\top J^{\perp} \si + h(y)^\top \si - \rc 2\ve{y}^2}\,dy}\\
    % &\le e^{1/2} \eta^d \cdot \fc{\max_{y\in \Grid} \sum_{\si\in \{\pm 1\}^d} \exp\pa{\rc 2 \si^\top J^{\perp} \si + h(y)^\top \si - \rc 2\ve{y}^2}\,dy }{\int_{[-L,L]^d}\sum_{\si\in \{\pm 1\}^d} \exp\pa{\rc 2 \si^\top J^{\perp} \si + h(y)^\top \si - \rc 2\ve{y}^2}\,dy}\cdot 2\pf{L}{\eta}^d\\
    % &\quad \text{because $\fc{Z(y)}{\wh Z(y)}\le 2$ for each $y$}\\
    % &\le e^{1/2} \eta^d \cdot \max_{y\in \Grid}\fc{ \sum_{\si\in \{\pm 1\}^d} \exp\pa{\rc 2 \si^\top J^{\perp} \si + h(y)^\top \si - \rc 2\ve{y}^2}\,dy }{\sum_{\si\in \{\pm 1\}^d}\int_{B(y)} \exp\pa{\rc 2 \si^\top J^{\perp} \si + h(y')^\top \si - \rc 2\ve{y'}^2}\,dy'}\cdot 2\pf{L}{\eta}^d \\
    % &\le e^{1/2} \eta^d \cdot \max_{y\in \Grid, \si\in \{\pm 1\}^d}\fc{ \exp\pa{\rc 2 \si^\top J^{\perp} \si + h(y)^\top \si - \rc 2\ve{y}^2}\,dy }{\int_{B(y)} \exp\pa{\rc 2 \si^\top J^{\perp} \si + h(y')^\top \si - \rc 2\ve{y'}^2}\,dy'}\cdot 2\pf{L}{\eta}^d \\ 
    % &\le e^{1/2} \eta^d \cdot \eta^{-d} e^{1/2} \cdot 2\pf{L}{\eta}^d = 2e\pf{L}{\eta}^d& \text{by~\eqref{e:vary-on-box}}\\
    \end{align*}
\end{enumerate}
Thus we can apply Lemma~\ref{l:is} with $C=16\exp(c\Tr(J_-)+2)\pf{2L}{\eta}^d$. Replacing $\ep$ with $\fc{\ep}{C}$, we get that the distribution restricted to $\{M\}\times \Grid \times \{\pm 1\}^n$ after running for $\Om\pa{\CP \log\pf{MC}{\ep}}$ steps is $\fc{\ep}{4C}$ close to $\Pst_M$ in TV-distance. By  Lemma~\ref{l:is}, an accepted sample will be $\fc{\ep}2$ close to $P_{M+1}$. Finally, because $L$ was chosen large enough so that $P(y\nin [-L,L]^d)\le \fc \ep 4$ as in~\eqref{e:cutoff}, we conclude that the marginal distribution of $\si$ is $\ep$-close to $P_{J,h}$. The expected number of trials until acceptance will be $O(CM) = O(n\exp(c\Tr(J_-))(2L/\eta)^d)$.
%_{J^\prl , \Jpall, h}
\end{proof}

%% file: equiv2.tex
\section{Interpreting the Hubbard-Stratonovich transform as as Gaussian mixture posterior} 
\label{app:equiv}
In this Appendix, we discuss at length the properties of the Hubbard-Stratonovich transform and its possible interpretation as a Gaussian mixture model posterior. For the most part (and unlike all of the other appendices in this paper) this discussion is pedagogical, though some simple formulas stated here are used elsewhere in the paper.

Throughout this section, we consider the case when $J$ is positive semi-definite (PSD). 
In this case, we can write $J=\rc nX^\top X$ for $X\in \R^{d\times n}$, for $d=\rank(J)\le n$. Let $x_1,\ldots, x_n$ be the columns of $X$; we will re-interpret the Hubbard-Stratonovich transform 
% as augmenting the probability space 
%$\pproj ^{\si,y}$ as 
as giving the posterior of a Gaussian mixture model after seeing samples $x_1,\ldots, x_n$. (The precise model is a very slight variant of the Gaussian mixture model described in the main text and applications sections.)
We consider the following augmented model, which is a density on $\{\pm 1\}^n\times \R^d$:
\begin{align}\label{e:augment}
p_{X,h}(\si,\mu) & = \rc{\Zjoint} \prod_{i=1}^n \exp\pa{-\rc 2 \ve{\si_i x_i-\mu}^2 + h_i\si_i}
\\ 
\text{where }
\Zjoint  & = \int_{\R^d} \sum_{\si\in \{\pm 1\}^n} \exp\pa{-\rc 2 \ve{\si_i x_i-\mu}^2 + h_i\si_i}\,d\mu.
%\propto \prod_{i=1}^n \exp\pa{\si_i \an{x_i,\mu}}.
\end{align}
(As we will see below, $\Zjoint$ does not depend on the choice of $X$.)
Note this can be interpreted as the posterior distribution for a Gaussian mixture model (with two components, symmetric around 0 with identity covariance) $p(x|\mu) \propto \exp\pa{-\rc 2\ve{x-\mu}^2} + \exp\pa{-\rc 2\ve{x+\mu}^2}$ with uniform prior on $\mu$
and prior on $\si$ given by $p_{\textup{prior}}(\si) \propto e^{\an{h,\si}}$, where $\si$ represents the class assignments (to the Gaussian with mean $\mu$ or mean $-\mu$).

We summarize the connection in this lemma. We will drop the subscripts $J,h$ when they are clear. %Note that when we write $\propto$, the constants of proportionality do not depend on the variables to the left of the conditioning. 
\begin{lem}\label{l:joint}
Consider the distribution $p_{X,h}(\si,\mu)$ in \eqref{e:augment} and let $J=\rc n X^\top X$.
The following hold:
\begin{enumerate}
\item
The marginal distribution of $\si$ is $p_{J,h}(\si)$ (in \eqref{e:ising}). 
\item
The marginal distribution on $\mu$ is
\begin{align*}
p(\mu) & \propto 
%\prodo in \pa{\exp\pa{-\rc 2 \ve{ x_i-\mu}^2} + \exp\pa{-\rc 2 \ve{ x_i+\mu}^2}}\propto
 e^{-\fc n2 \ve{\mu}^2} \prodo in \cosh(\an{x_i,\mu}+h_i).
\end{align*}
\item
The conditional distribution of $\si$ given $\mu$ is a product distribution,
\begin{align*}
p(\si|\mu) &\propto \prod_{i=1}^n \exp\pa{\si_i (\an{x_i,\mu}+h_i)}.
\end{align*}
\item
The conditional distribution of $\mu$ given $\si$ is a Gaussian distribution,
\begin{align*}
p(\mu|\si) &= \pf{n}{2\pi}^{\fc n2} \exp\pa{-\fc n2 \ve{\mu - \fc{\sumo in \si_i x_i}n}^2}.
\end{align*}
\item
The partition functions are related via
\begin{align*}
\Zjoint  &= \pf{2\pi}{n}^{n/2} \exp\pa{-\fc n2 \Tr(J)}\ZJh .
\end{align*}
\end{enumerate}
\end{lem}
As a consequence, to sample from $p(\si)$, it suffices to sample $\mu$ from the above distribution, and then sample $\mu$ conditional on $\mu$ (which is immediate).

We calculate the Hessian of $-\ln p(\mu)$:
\begin{align*}
-\nb^2 \ln p(\mu) &= nI - \sumo in x_ix_i^\top + \sumo in (1-\sech^2(\an{x_i,\mu} + h_i)) x_ix_i^\top.
\end{align*}
Note that this is convex (and hence $p(\mu)$ is log-concave) when $J =\rc n \sumo in x_ix_i^\top\preceq I$. This observation can be used to infer an efficient sampling algorithm for $p_{J,h}$ by first drawing a sample from $p(\mu)$ (using algorithms for log-concave sampling such as Langevin dynamics~\citep{durmus2019analysis}) and then drawing from $p(\si|\mu)$, as observed in \cite{bauerschmidt2019very}. This gives an alternative algorithm to the Glauber dynamics (which mix rapidly under the same assumption \citep{anari2021aentropic}), albeit one which is not as fast.
%% the argument which works doesn't use this decomposition.
%, or by inferring fast mixing for Glauber dynamics through a decomposition theorem~\citep{bauerschmidt2019very,eldan2020spectral}.
%two-scale functional inequality

We note that our decomposition is similar, but slightly different from the decomposition in~\cite{bauerschmidt2019very}. 
Both approaches decompose $p_{J,h}$ as a log-concave mixture of product distributions when $J\preceq I$. 
Our approach has the advantage that when $J$ has a few large eigenvalues (eigenvalues greater than 1), the distribution on $\mu$ is still log-concave in the other directions. We note the log-concave decomposition technique was used extensively in analysis of the Hopfield model \citep{bovier2012mathematical,talagrand2010mean}.

\begin{proof}
\begin{enumerate}
\item
The marginal distribution of $\si$ is 
\begin{align}
\nonumber
&\rc{\Zjoint }\int_{\R^d} \prodo in \exp\pa{ -\rc 2\ve{\si_ix_i-\mu}^2+h_i\si_i}\,d\mu
=\rc{\Zjoint } \int_{\R^d} \exp\pa{-\rc 2 \sumo in \ve{\si_i x_i-\mu}^2+h_i\si_i}\,d\mu\\
\nonumber
&= \rc{\Zjoint } \int_{\R^d} \exp\pa{-\fc n2 \ve{\mu - \fc{\sumo in \si_i x_i}{n}}^2 + \rc{2n} \sumo{i,j}n \si_i x_i x_j^\top \si_j - \rc 2 \sumo in \ve{x_i}^2 + \an{h,\si}}\,d\mu\\
& = \rc{\Zjoint }\pf{2\pi}n^{n/2} \exp\pa{-\rc 2 \ve{X}_F^2} \exp\pa{\rc 2 \si^\top \pf{XX^\top}n\si+ \an{h,\si}},
\label{e:comp-nc}
\end{align}
where the last line uses the fact that the integral of $\exp\pa{-\fc n2 \ve{\mu-\mu_0}^2}$ is a fixed normalizing constant, for any $\mu_0$. Finally, we use $J=\rc n XX^\top$.
\item This follows from factoring the product,
\begin{align*}
p(\mu) &\propto \sum_{\si\in \{\pm 1\}^n} \prodo in \exp\pa{-\rc 2 \ve{\si_i x_i-\mu}^2+h_i\si_i} \propto \prodo in \sum_{\si_i=\pm 1}\exp\pa{-\rc 2 \ve{\si_ix_i-\mu}^2+h_i\si_i}\\
&\propto
e^{-\rc 2 \ve{\mu}^2} \prodo in \sum_{\si_i=\pm 1} e^{\si_i (\an{x_i,\mu}+h_i)} \propto  e^{-\fc n2 \ve{\mu}^2} \prodo in \cosh(\an{x_i,\mu}+h_i).
\end{align*}
\item[3--4.] These follow directly by noting $p(\si|\mu)\propto p(\si,\mu)$ for fixed $\mu$, and $p(\mu|\si) \propto p(\si,\mu)$ for fixed $\si$.
\item[5.] This follows from comparing normalizing constants in~\eqref{e:comp-nc}.
\end{enumerate}
\end{proof}

Lemma~\ref{l:joint} gives a decomposition of $p_{J,h}$ into a mixture of product distributions $p_{J,h}(\si) = \int_{\R^d} p(\si|\mu) p(\mu)\,d\mu$. We can instead only condition on the projection of $\mu$ to a rank-$d$ subspace $V$ and obtain a decomposition in terms of rank-$(n-d)$ Ising models. We will choose the rank-$d$ subspace to contain the eigenvectors of $J$ with large eigenvalue.

We define the distribution $p_{X,h,V}(\si,\mu^\prl, \mu^\perp)$ on $\{\pm 1\}^n\times V\times V^\perp$ by $p_{X,h,V}(\si,\mu^\prl, \mu^\perp) = p_{X,h}(\si,\mu^\prl+\mu^\perp)$.
%\hlnote{Fix this to be consistent with preliminaries section.}

%The decomposition allows us to consider the following decomposition
\begin{lem}\label{l:joint-prl}
Consider the distribution $p(\si,\mu^\prl, \mu^\perp)$.
Let $P^\prl$ and $P^\perp$ be the projections onto $V$ and $V^\perp$, respectively and let $J^\prl = \rc n X^\top P^\perp X$, $J^\perp = J-J^\prl$. 
%$J^\perp = \rc n X^\top P^\perp X$.
\begin{enumerate}
\item
The joint distribution of $(\si,\mu^\prl)$ is given by
\begin{multline*}
    p(\si,\mu^\prl) 
     = \rc{\Zjoint }
\pf{2\pi}{n}^{(n-d)/2}\exp\pa{-\fc n2 \Tr(J)}\\
    \cdot \exp\pa{\rc{2} \an{\si, J^\perp \si}
    + \an{X^\top P^\prl \mu+h ,\si} - \fc n2 \ve{\mu^\prl}^2}.
\end{multline*}
    \item The distribution of $\si$ given $\mu^\prl$ is
\begin{align*}
p(\si|\mu^\prl)&
= 
p_{J^\perp , h + X^\top \mu^\prl}(\si) 
\propto 
\exp\pa{\isingexp{J^\perp}{h+X^\top \mu^\prl}}.
\end{align*}
\item 
The distribution of $\mu^\prl$ given $\si$ is Gaussian,
\begin{align*}
p(\mu^\prl|\si) &=\pf{n}{2\pi}^{\fc d2} \exp\pa{-\fc n2 \ve{\mu^\prl - \fc{\sumo in \si_i P^\prl x_i}n}^2}.
\end{align*}
\item
% Letting
% \[
% \Zproj := \int_{\mu^\prl \in V}
% \sum_{\si\in \{\pm 1\}^n}
% \exp\pa{\rc{2} \an{\si, J^\perp\si}
%     + \sumo in \pa{-\rc 2 \ve{\si_i P^\prl x_i - \mu^\prl}^2 + h_i\si_i}}\,d\mu^\prl,
% \]
Let $\Zproj := \int_{\mu^\prl \in V}
\Zproj(\mu^\prl) \,d\mu^\prl$ where 
\begin{align}
\label{e:ZXhVmu}
\Zproj (\mu^\prl) :&= 
Z_{J^\perp, X^\top P^\prl \mu + h}\exp\pa{-\fc n2 \ve{\mu^\prl}^2\,d\mu^\prl} \\
&= 
\sum_{\si\in \{\pm 1\}^n}\exp\pa{\rc 2 \an{\si, J^\perp \si}+\an{X^\top P^\prl \mu + h,\si} -\fc n2 \ve{\mu^\prl}^2}.
\end{align}
Then we have
\begin{align}
\Zjoint  &=
\pf{2\pi}{n}^{d/2}\exp\pa{-\fc n2 \Tr(J^\perp)} \Zproj .
\end{align}
\end{enumerate}
\end{lem}
\begin{proof}
\begin{enumerate}
    \item We integrate $p(\si,\mu^\prl, \mu^\perp)$ along $V^\perp$ and complete the square in $\mu^\perp$; integrating gives a $\pf{2\pi}n^{(n-d)/2}$ normalizing constant:
%https://tex.stackexchange.com/questions/29730/how-can-i-use-multline-within-an-align-environment
%https://tex.stackexchange.com/questions/64998/spacing-of-multline-inside-align-i-e-multlined
\begin{align*}
p(\si,\mu^\prl)
& = \int_{\mu^\perp \in V^\perp}p(\si,\mu^\prl,\mu^\perp) \,d\mu^\perp \\
&= \rc{\Zjoint } \int_{\mu^\perp \in V^\perp} \exp\pa{-\rc 2\sumo in \pa{ \ve{\si_iP^{\perp} x_i-\mu^\perp }^2 + \ve{\si_iP^{\prl} x_i-\mu^\prl }^2} + \an{h,\si}}\,d\mu_\perp \\
 & =\!\begin{multlined}[t] \rc{\Zjoint } \int_{\mu^\perp \in V^\perp}  \exp\Bigg(-\rc 2\Bigg( n\ve{\mu^\perp}^2 - \sumo in \si_i \an{\mu^\perp, P^\perp x_i} + \sumo in \ve{P^\perp x_i}^2 \\
 + \sumo in \ve{\si_iP^{\prl} x_i-\mu^\prl }^2\Bigg) + \an{h,\si}\Bigg)\,d\mu_\perp \end{multlined}\\
 & = \!\begin{multlined}[t] \rc{\Zjoint } \int_{\mu^\perp \in V^\perp}
 \exp\Bigg( - \fc n2 \ve{\mu^\perp - \rc n \sumo in \si_i P^\perp x_i}^2 + \fc{1}{2n} \ve{\sumo in \si_i P^\perp x_i}\\  - \rc 2 \sumo in \ve{P^\perp x_i}^2 -\rc 2\sumo in  \pa{\ve{P^\prl x_i}^2 - \an{X^\top P^\prl\mu, \si} + \ve{\mu^\prl}^2} + \an{h,\si}\Bigg)\,d\mu_\perp \end{multlined}\\
 &= \rc{\Zjoint }\pf{2\pi}{n}^{(n-d)/2} \exp\pa{ \rc 2 \an{\si, \fc{X^\top P^\perp X}n\si} - \rc 2 \ve{X}_F^2+ \an{X^\top P^\prl \mu+h ,\si} - \fc n2 \ve{\mu^\prl}^2}
%   &= \pf{2\pi}n^{n/2} \exp\pa{- \rc 2 \ve{P^\perp X}_F^2} \Zproj .
\end{align*}
Finally, we rewrite in terms of $J^\perp$ by using $J^\perp = \rc n X^\top P^\perp X$.
\item 
This follows from fixing $\mu^\prl$ in the joint probability density and expanding.
\item
This follows from fixing $\si$ in the joint density, expanding, and completing the square in $\mu^\prl$. 
\item This follows from setting the integral of the joint density equal to 1.
\end{enumerate}
\end{proof}

Finally, we note that although the interpretation as a Gaussian mixture posterior only makes sense when $J$ is positive semi-definite, the decomposition still works for general symmetric $J$, as we can multiply the distribution by $\exp\pa{-\rc 2\an{\si, J_-\si}}$.
We note that combining Lemma~\ref{l:joint}, part 5, with Lemma~\ref{l:joint-prl}, part 4, gives us Theorem~\ref{t:hs} in the PSD case.

%% file: appendix.tex
\section{Technical lemmas for partition function estimation and sampling}\label{app:technical}

In this section, we collect some technical lemmas we will need for analyzing our algorithms for partition function estimation and sampling.

\subsection{Simulated annealing}

For partition function estimation, we use the following lemma, which roughly says that when the variance of some random variables are close to 1, then the variance is additive under multiplication.
\begin{lem}[{\cite[Lemma B.2]{ge2020estimating}, cf. \cite{dyer1991computing}}]\label{lem:prodvar}\label{l:prod}
    Let $Y_\ell$, $\ell = 1, \ldots, M$ be independent variables and let $\ol{Y}_\ell = \E Y_\ell$. Assume there exists $\eta > 0$  such that $\eta M \leq \frac{1}{5}$ and
    \begin{equation*}
        \E Y_\ell^2 \leq (1 + \eta) \ol{Y}_\ell^2,
    \end{equation*}
    then for any $\ep > 0$
    \begin{equation*}
        \mathbb{P}\Biggl( \frac{ \ab{ Y_1 \cdots Y_M - \ol{Y}_1 \cdots \ol{Y}_M}}{\ol{Y}_1 \cdots \ol{Y}_M} \geq \frac{\ep}{2} \Biggr) \leq \frac{5 \eta M}{\ep^2}.
    \end{equation*}
\end{lem}

\begin{proof}[Proof of Lemma~\ref{l:sa}]
Let $Y_\ell = \E_{P_\ell} \fc{q_\ell}{p_\ell} = \fc{\int_{\Om} q_{\ell+1}}{\int_{\Om}q_\ell}$. 
By  Lemma~\ref{l:prod} with $\eta =\fc{\si^2}N$,
\begin{align}
\label{e:sched-var}
\Pj\pa{%\fc{\ab{\prodo i{M} Y_i - \prodo i{M}\ol Y_i}}{\prodo i{M}\ol Y_i}
\prodo \ell M Y_\ell 
\nin [e^{\ep/2},e^{\ep/2}]\cdot  \prodo \ell M \ol Y_\ell
} 
&\le 
\Pj\pa{\fc{\ab{\prodo \ell{M} Y_\ell - \prodo \ell{M}\ol Y_\ell}}{\prodo \ell{M}\ol Y_\ell} \ge \fc\ep 4}
\le \fc{80\eta M}{\ep^2} \le \rc 4.
%20 here?
\end{align}

Now we consider the bias. We have
\[
%\ab{\E_{p_i} g_i(x) - \E_{\wt p_i}} \le 
\E_{\wt P_\ell} g_\ell(x) \in 
\ba{ 1- \fc{\ep}{4M}, 1+ \fc{\ep}{4M}}\cdot \E_{P_\ell} g_\ell(x)
\subeq [e^{-\fc{\ep}{2M}}, e^{\fc{\ep}{2M}}] \cdot \E_{P_\ell} g_\ell(x).
\]
Taking a product, we obtain %that with probability at least $\fc 34$,
\begin{align}
\label{e:sched-bias}
    Z_1\prodo \ell{M} \ol Y_\ell &\in \ba{e^{-\eph}, e^{\eph}}\cdot Z.
\end{align}
Putting together~\eqref{e:sched-var} and~\eqref{e:sched-bias}, 
we obtain that for any $r$, 
\[
\Pj\pa{\wh Z^r \nin [e^{-\ep}, e^\ep] Z}\le \fc 14.
\]
The algorithm takes the median in order to boost this probability. As the median of $R$ independent runs, $\wh Z$ will fail to be contained in $[e^{-\ep}, e^\ep]\cdot  Z$ only if at least half of the $\wh Z^r$'s fail to be contained in $[e^{-\ep}, e^\ep]\cdot  Z$. By the Chernoff-Hoeffding bound, this happens with probability at most $\de$ when $R\ge 32 \log\prc \de$. 
\end{proof}

%\subsection{Inequalities}

% \begin{lem}\label{l:overlap-chi}
% Let $P$ be a probability measure and $Q$ be a nonnegative measure on $\Om$.
% Define the measure $R$ by $R=\min\bc{\dd{Q}{P}, 1}P$. (If $p,q$ are the density functions of $P,Q$, then the density function of $R$ is simply $r(x)=\min\{p(x),q(x)\}$.) %Let $\de$ be the overlap $\de=R(\Om)$, and 
% Let $\wt R=\fc{R}{R(\Om)}=\fc{R}{1-\TV(P,Q)}$ be the normalized overlap measure.
% %Let $p(x)$ and $q(x)$ be such that $\int_{\Om}p(x)\dx=1$. 
% %Let $\de$ be the overlap: $\de=\int \min\{p(x),q(x)\}\dx$. 
% Then
% \begin{align}
% \chi^2\pa{R||P} &\le  \min\bc{\fc{2d_{\TV}(P,Q)}{\pa{1-d_{\TV}(P,Q)}^2}, \rc{1-d_{\TV}(P,Q)}}.
% \end{align}
% \end{lem}

% \begin{proof}
% We calculate
% \begin{align*}
% \chi^2\pa{R||P} 
% &= \int_{\Om} \pa{\fc{\min\bc{1,\dd QP}}{1-\TV(P,Q)}-1}^2\,P(dx)\\
% &\le \int_{\Om} \pa{\fc{\min\bc{1,\dd QP}-(1-\TV(P,Q))}{1-\TV(P,Q)}}^2\,P(dx)\\
% &\le \rc{(1-\TV(P,Q))^2} \int \ab{\pa{\min\bc{1,\dd QP}-1} + \TV(P,Q)}^2\,P(dx)\\
% &\le \rc{(1-\TV(P,Q))^2} \int \ab{\pa{\min\bc{1,\dd QP}-1} + \TV(P,Q)}\,P(dx)\\
% &\le \fc{2\TV(P,Q)}{\pa{1-\TV(P,Q)^2}}.
% \end{align*}
% Note we also have the following. Since $\int\min\bc{1,\dd QP}\,P(dx)=1-\TV(P,Q)$, 
% \begin{align*}
% &\quad 
%  \rc{(1-\TV(P,Q))^2} \int\pa{\min\bc{1,\dd QP}-(1-\TV(P,Q))}^2 \,P(dx)\\
%  &\le \rc{(1-\TV(P,Q))^2} \int\pa{\min\bc{1,\dd QP}}^2\,P(dx)\\
% &\le \rc{(1-\TV(P,Q))^2} \int\pa{\min\bc{1,\dd QP}}\,P(dx)\\
% &\le \rc{1-\TV(P,Q)}.
% \end{align*}
% \end{proof}

\subsection{Rejection sampling}

The following bounds the TV-error and expected running time for rejection sampling, given an inexact oracle for the proposal distribution.

\begin{lem}\label{l:is}
Suppose that $P$ and $Q$ are probability measures on $\Om$ such that $\dd PQ\le C$ everywhere. Suppose we have an oracle which gives samples from $\wt Q$, with $d_{\TV}(\wt Q, Q)\le \fc{\ep}{2C}$.
Consider the following rejection sampling algorithm: draw $x\sim \wt Q$, and accept with probability $\rc C\dd{P}{Q}(x)$; otherwise repeat the process. Let $\wt P$ be the resulting measure. Then $d_{\TV}(\wt P, P)\le\ep$, and the number of oracle calls is a geometric random variable with success probability at least $\rc{2C}$ (and hence expected value at most $2C$).
\end{lem}

\begin{proof}
Let $A\subeq \Om$ be measurable.
First, we note that $\wt P(A) = \fc{\int_A \dd PQ d\wt Q}{\int_\Om \dd PQ d\wt Q}$. To calculate $d_{\TV}(\wt P,P)$, we break up the difference as
\begin{align*}
    \wt P(A) - P(A) 
    &= \pa{\fc{\int_A \dd PQ \,d\wt Q}{\int_{\Om} \dd PQ \,d\wt Q} - \int_A \dd PQ\, d\wt Q}
    + \pa{\int_A \dd PQ\, d\wt Q - \int_A \dd PQ\, dQ}\\
    &\le \pa{\fc{\int_A \dd PQ \,d\wt Q}{\int_{\Om} \dd PQ \,d\wt Q}\pa{1-\int_{\Om} \dd PQ \,d\wt Q}}  + \pa{\int_A \dd PQ\, d\wt Q - \int_A \dd PQ\, dQ}
\end{align*}
Next note that 
\begin{align*}%\label{e:is1}
\ab{\int_{\Om} \dd PQ \,d\wt Q - 1} \le 
    \ab{\int_{\Om} \dd PQ \,d\wt Q - \int_\Om \dd PQ \,dQ}\le Cd_{\TV} (Q, \wt Q) \le \fc{\ep}2.
\end{align*}
Hence,
\begin{align*}
    |\wt P(A) - P(A) | &\le
    \ab{\int_{\Om} \dd PQ \,d\wt Q - 1} + d_{\TV}(\wt Q, Q) \ve{\dd PQ}_{\iy} 
    \le \fc{\ep}2 + \fc \ep{2C}C = \ep,
\end{align*}
so $d_{\TV}(\wt P , P)\le \ep$. %By~\eqref{e:is1}, 
Finally, we check that the acceptance probability is \[\int_\Om \rc C \dd PQ d\wt Q \ge \int_\Om \rc C \dd PQ dQ - \rc C \cdot C d_{\TV}(Q, \wt Q) \ge  \rc C - \fc{\ep}{2C}\ge \rc{2C}.\]
\end{proof}

\subsection{Spectral gap of a projected chain}

We use the following to bound the Poincar\'e constant of the projected Markov chain arising in the analysis of simulated tempering.  A similar analysis appears in the proof in~\cite{ge2018simulated}.

%This lemma is from [CITE]. \hlnote{I am trying to get the paper out on arxiv but if not I will just import the proof. We just need the Poincar\'e part.}
\begin{lem}\label{l:proj}
Let $S$ be a countable set. 
Consider a reversible Markov chain on $[L]\times S$ with stationary distribution $P$ and transition kernel $T$ satisfying the following conditions. Let $P_\ell(j) = P((\ell,j))/P(\{\ell\}\times S)$.
%\begin{enumerate}
%\item (Transitions at highest temperature) $T((1,i),(1,j))=\Om(P_1(j))$.
%\item (Transitions between adjacent temperature) $T((\ell,j),(\ell\pm 1,j))=\Om(\min\{\fc{P((\ell \pm 1, j))}{P((\ell , j))}, 1\})$. 
%\item (Bounded bottleneck ratio) For $k<\ell$, $\fc{P_k(j)}{P_\ell(j)}\ge \ga$. 
%\end{enumerate}
%Then $ \Phi = \Om\pf{\ga}{L}$, $C_P = O\pf{L^2}{\ga^2}$, and $\CLS =O\pa{ \fc{L^2}{\ga^2} \ln (1/p_{\min} )}$.
%Let $T$ be a transition matrix on $[L]\times [n]$ with 
%non-negative matrix indexed by $([n]\times \{0,1,\ldots, L\})^2$ with %$L_{(i,\ell),(i,\ell)}=0$ and 
\begin{enumerate}
\item (Bounded bottleneck ratio) For $k<\ell$, $\fc{P_k(j)}{P_\ell(j)}\ge \ga$. 
\item (Transitions at highest temperature and between adjacent temperatures)
We have
\begin{align*}
T((\ell_1,i_1),(\ell_2,i_2))
&\ge \begin{cases}
\fc{P_1(i_2)}{D_{\textup{high}}}, & \ell_1=\ell_2=1, \quad i_1\ne i_2\\
\rc {2D_{\textup{adj}}}\min\bc{\fc{P((\ell \pm 1, i_1))}{P((\ell , i_1))}, 1} , & i_1=i_2,\quad  \ell_1\ne L, \ell_2=\ell_1\pm 1
\end{cases}
\end{align*}
\item (Lower bound of probability for each level) For each $\ell$, $P(\{\ell\}\times S) \ge \fc rL$. 
\end{enumerate}
%(For a Markov chain, this gives the transition probabilities; for a Markov process, this is the generator.)
%In the case of a Markov chain, additionally assume that the chain is lazy, that is, $T((\ell,i),(\ell,i))\ge \rc 2$ for any $(\ell,i)$. 
Then the following hold.
\begin{enumerate}
\item
(Cheeger constant) The Cheeger constant %$\Phi = \min_{A\subeq\Om, P(A)\le \rc 2} \fc{\sum_{\om \in A}p(x)T(x,A^c)}{P(A)}$ 
satisfies $\Phi \ge \fc{\ga r}{2L\max\{D_{\textup{high}}, D_{\textup{adj}}\}}$.
\item 
(Poincar\'e constant) The associated Dirichlet form satisfies a Poincar\'e inequality with constant $\CP\le \fc{8L^2\max\{D_{\textup{high}}, D_{\textup{adj}}\}^2}{\ga^2 r^2}$.
%\item
%(Log-Sobolev constant) The associated Dirichlet form satisfies a log-Sobolev inequality with constant $\CLS \le\fc{16L^2\max\{D_{\textup{high}}, D_{\textup{adj}}\}^2(1+\ln (1/p_{\min} ))}{\ga^2 r^2}$, where $p_{\min} = \min_{(\ell,j)}P((\ell,j))$.
\end{enumerate}
\end{lem}

\begin{proof}
Let $Q(x,B)$ denote $P(x)T(x,B)$ and $Q(A,B)$ denote $\sum_{x\in A} P(x)T(x,B)$. Note $Q(A,B) = Q(B,A)$ by reversibility. 
Let $A_\ell$ denote the sets such that $A=\cupo \ell L \{\ell\}\times A_\ell$, i.e., $A_\ell$ is the $\ell$th layer of $A$.

To prove the bound on the Cheeger constant, for each $A$, it suffices to bound either $\fc{Q(A,A^c)}{P(A)}$ or $\fc{Q(A^c,A)}{P(A^c)}$. Without loss of generality, we suppose that $P_1(A_1)\le \rc 2$. 
For each $j$, let $\ell_j$ denote the smallest $\ell$ such that $(\ell,j)\in A$. To lower bound $Q(A,A^c)$, we consider the contributions from $n$ such that $\ell_j>1$ and $\ell_j=1$ separately. 
\begin{enumerate}
\item
$\ell_j>1$: We have 
\begin{align*}
Q((\ell_j,j), A^c)&\ge 
P((\ell_j,j))  T((\ell_j,j), (\ell_{j}-1,j))\\
&\ge  P((\ell_j,j)) \rc{2D_{\textup{adj}}} \min\bc{\fc{P((\ell_j - 1, j))}{P((\ell_j , j))}, 1}\\
&= \rc{2D_{\textup{adj}}} \min\{P((\ell_j - 1, j)),P((\ell_j , j))\}\\
&\ge \fc{\ga r}{2LD_{\textup{adj}}} P([\ell_j,L] \times \{j\}).
\end{align*}
\item 
$\ell_j=1$: Note $A_1=\set{j}{\ell_j=1}$. We will bound $Q(\{1\}\times A_1, A^c)$ by looking at transitions within $\{1\}\times S$. We have
\begin{align*}
Q(\{1\}\times A_1, A^c) &\ge \sum_{j\in A_1} P((1,j)) T((1,j), \{1\}\times A_1^c)\\
&\ge \sum_{j\in A_1} P((1,j)) \fc{P(\{1\}\times A_1^c)}{D_{\textup{high}}}\\
&\ge \rc{2D_{\textup{high}}} \sum_{j\in A_1} P((1,j)) = \rc{2D_{\textup{high}}} P(\{1\}\times A_1)\\
&\ge \fc{\ga r}{2LD_{\textup{high}}} P([L]\times A_1).
\end{align*}
\end{enumerate}
Adding the two parts,
\begin{align*}
Q(A,A^c) &\ge \fc{\ga r}{2L\max\{D_{\textup{adj}},D_{\textup{high}}\}} 
P\pa{\pa{\bigcup_{j:\ell_j>1}[\ell_j,L]\times \{j\}} \cup ([L]\times A_1)}\\
&\ge \fc{\ga r}{2L\max\{D_{\textup{adj}},D_{\textup{high}}\}}  P(A).
\end{align*}

The bound on the Poincar\'e constant follows immediately from Cheeger's inequality: the spectral gap of the chain is at least $\rc 2 \Phi^2$, and the Poincar\'e constant is the inverse of the spectral gap.
\end{proof}

\section{Additional material related to examples}\label{a:example}
We give here the derivation of the posterior for the contextual SBM. Because we chose consistent notations between problems, the derivation of the posterior for the Gaussian mixture model is simply the special case of this argument where $\lambda = 0$ (so there is no graph/spiked Wigner information).
%\fnote{fred: clean this up}

\paragraph{Posterior derivation in contextual SBM.}
Under the Gaussian contextual stochastic block model, we have 
\begin{align*} 
p(A,B \mid u,v) 
&\propto \exp\left(-\frac{n}{4} \ve{\frac{\lambda}{n} vv^\top  - A}_F^2 - \frac{p}{2} \ve{\sqrt{\frac{\mu}{n}} vu^\top  - B}_F^2\right) \\
&\propto \exp\left(\frac{\lambda}{2} \langle vv^\top , A \rangle - \fc n4\|A\|_F^2 + p\sqrt{\mu/n} \langle vu^\top , B \rangle -  \frac{p}{2} \|B\|_F^2 - \frac{p}{2} \mu \|u\|^2 \right)
\end{align*}
(note we dropped the term $\|vv^\top \|_F^2$ since it is a constant)
and so
\begin{align*}
p(u,v \mid A,B) &= p(A,B \mid u,v) p(u,v)/p(A,B)\\ &\propto \exp\left(\frac{\lambda}{2} \langle vv^\top , A \rangle + p\sqrt{\mu/n} \langle B^\top  v, u \rangle - \frac{p}{2}(1 + \mu) \|u\|^2 \right).
\end{align*}
Integrating over $u$, we have that the posterior distribution is
\begin{align*}
p(v \mid A,B) 
&\propto \int \exp\left(\frac{\lambda}{2} \langle vv^\top , A \rangle + p\sqrt{\mu/n} \langle B^\top  v, u \rangle - \frac{p}{2} (1 + \mu)\|u\|^2 \right) du \\
&= \int \exp\left(\frac{\lambda}{2} \langle vv^\top , A \rangle - \frac{p}{2}(1 + \mu) \ve{u - \rc{1 + \mu}\sfc{\mu}{n}B^\top v}^2 + \frac{p\mu}{2n(1 + \mu)}\|B^\top  v\|^2 \right) du \\
&\propto  \exp\left(\frac{\lambda}{2} \langle vv^\top , A \rangle + \frac{p\mu}{2n(1 + \mu)} \| B^\top  v \|_2^2  \right) \\
&\propto  \exp\left(\frac{\lambda}{2} \langle vv^\top , A \rangle + \frac{p\mu}{2n(1 + \mu)}\langle vv^\top , BB^\top  \rangle \right)
\end{align*}
\iffalse
\[ p(A,B,u \mid v) \propto \exp\left(\frac{\lambda}{2 n} \langle vv^\top , A \rangle - \|A\|_F^2/4 + \sqrt{\mu/n} \langle vu^\top , B \rangle -  \frac{1}{2} \|B\|_F^2 - p\|u\|_2^2/2\right) \]
\fi
%b_i \sim N\left(\sqrt{\frac{\mu}{n}} v_i u, I_p/p\right), \qquad u \sim N(0,I_p/p).  \]
This is an Ising model without external field.

%% file: hardness.tex
\section{Computational hardness of sampling from rank-one models with large spike}\label{s:hardness}
Using the subset sum/number partitioning problem, we will show that sampling and (even crudely) approximating $\log Z$ from negative-definite rank-one models is $\mathsf{NP}$-hard. 
The $\mathsf{NP}$-hard problem we start with is given integers $a_1,\ldots,a_n$, determining whether there exists a partitioning into two sets such that the sum is equal. Equivalently, we seek to determine if there exists a sign vector $\sigma \in \{\pm 1\}^n$ such that
\[ \sum_i a_i \sigma_i = 0. \]
%We remind the reader that while there are other formulations of subset sum (e.g. where the numbers are nonnegative and the target sum is nonzero), there are simple reductions to this version of the problem (recenter
%Versions of this problem with random $a_i$ have been extensively studied from a statistical physics perspective, see the discussion
This is not the first time this problem is connected to statistical physics---see e.g., discussion in \cite{borgs2001phase,gamarnik2021algorithmic}.

\begin{thm}\label{thm:hardness}
Let $\beta \ge 1$ be arbitrary and fixed.
For any $a = (a_1,\ldots,a_n) \in \mathbb{Z}^n$, define the Ising model with probability mass function $P_a : \{\pm 1\}^n \to [0,1]$ given by
\[ P_a(\sigma) = \frac{1}{Z} \exp\left(-\beta n \langle a, \sigma \rangle^2\right) \]
%where the interaction matrix $J$ is rank one and negative definite, then $NP = RP$.
If there exists a polynomial time randomized algorithm to approximately sample within TV distance $1/2$ from Ising models of this form for any $a_1,\ldots,a_n$,
then $\mathsf{NP} = \mathsf{RP}$.
Furthermore, for $\be\ge 2\log(2)$, it is $\mathsf{NP}$-hard to 
approximate the log partition function/free energy $\log Z$ of such a
model within an additive error of $\frac{\beta n}{2}$, and under the Exponential Time Hypothesis (ETH), it is impossible to do so in subexponential time in the presence of an external field $b \in \mathbb{Z}^n$, i.e., for models of the form
\[ P_{a,h}(\sigma) = \frac{1}{Z} \exp\left(-\beta n \langle a, \sigma \rangle^2 + \langle b, \sigma \rangle \right). \]
\end{thm}
\begin{proof}
Let $a_1,\ldots,a_n$ be an instance of the number partitioning problem. 
Consider the Ising model with probability mass function $P_a : \{\pm 1\}^n \to [0,1]$ given by for $\beta \ge 1$
\[ P_a(\sigma) = \frac{1}{Z} \exp\left(-\beta n \langle a, \sigma \rangle^2\right), \]
where $Z$ is the normalizing constant (partition function) so that the distribution has normalizing constant $1$.
Note that this is an Ising model with interaction matrix $-2\beta n aa^T$, which is negative definite and rank one as promised. If there exists at least one solution $\sum_i a_i \sigma_i = 0$ then 
\[ \Pr_{\sigma \sim P}\left(\sum_i a_i \sigma_i \ne 0\right) = \frac{\sum_{\sigma : \sum_i a_i \sigma_i \ne 0} e^{-\beta n \langle a, \sigma \rangle^2}}{\sum_{\sigma \in \{\pm 1\}^n} e^{-\beta n \langle a, \sigma \rangle^2}} \le 2^ne^{-\beta n}  \]
where we used that because the $a_i$ are integers, if $\sum_i a_i \sigma_i \ne 0$ then $\langle a, \sigma \rangle^2 \ge 1$, and also that if there exists a solution $\sum_i a_i \sigma_i = 0$ then the denominator is at least $1$.
Thus, except with exponentially small probability in $n$, a sample from $P$ will be a solution to the subset sum problem. In particular, it follows that a polynomial time (approximate) sampling algorithm implies $\mathsf{NP} = \mathsf{RP}$. %\hlnote{Citation?}

Similarly, observe that if there exists a solution to the subset sum instance then $\log Z \ge 0$ whereas if there does not exist a solution, then $\log Z \le n[\log(2) - \beta] < -\frac{\beta n}{2}$, which establishes the $\mathsf{NP}$-hardness of approximating $\log Z$. The last statement in the Theorem follows because solving subset sum in time $2^{o(n)}$ is known to be ETH-hard (see discussion in \cite{abboud2022seth}), and the general subset problem (deciding if there exists $\sigma$ so that $\sum_i a_i \sigma_i = b$) can be directly encoded as minimizing
\[ \left(\langle a, \sigma \rangle - b\right)^2 = \langle a, \sigma \rangle^2 - 2b\langle a, \sigma \rangle + b^2, \]
which by the same argument as above implies that approximating $\log Z$ for the distribution $P_{a,h}$ with $h = 2b\beta na$ is ETH-hard.
%\fnote{may technically need to allow external field for this}
\end{proof}
%\fnote{this is only weak np-hardness. would be interesting if strong np-hardness possible. Might be possible by reduction to 3-partition or related problem? Create multiple variables for each number? We can solve multiple equations of the form $\langle a_1, \sigma \rangle = b_1, \langle a_2, \sigma \rangle = b_2, \langle a_3, \sigma \rangle = b_3$? But dynamic programming might work on this if numbers are all small.}

%Note that as always, the interaction matrix can be made positive semidefinite by adding a sufficiently large copy of the identity matrix, in which case the interaction matrix will have $n - 1$ eigenvalues larger than $1$, i.e. it will no longer be low rank. So the hardness result does not conflict with our sampling result.

%\fnote{can strengthen this result by using that under ETH, subset sum requires essentially brute force time, see e.g. \url{https://arxiv.org/pdf/1704.04546.pdf} and references}